\documentclass[11pt]{article}
\usepackage[margin=1in]{geometry}
\usepackage{amssymb,amsmath,amsthm,amsfonts}
\usepackage{mathrsfs}
\usepackage[dvipsnames,svgnames,x11names]{xcolor}
\usepackage{soul}

\usepackage{todonotes}
\usepackage{microtype,xspace}
\usepackage{graphicx,tikz}
\usepackage{framed}
\usepackage{url}
\usepackage{authblk}

\newtheorem{theorem}{Theorem}[section]
\newtheorem{lemma}[theorem]{Lemma}
\newtheorem{claim}[theorem]{Claim}
\newtheorem{corollary}[theorem]{Corollary}
\newtheorem{proposition}[theorem]{Proposition}
\newtheorem{observation}[theorem]{Observation}
\theoremstyle{definition}
\newtheorem{definition}[theorem]{Definition}

\newcommand{\NCTD}[1]{#1-NCTD}
\newcommand{\NCTDlong}[1]{nice #1-clique tree decomposition}
\newcommand{\NCPD}[1]{#1-NCPD}
\newcommand{\NCPDlong}[1]{nice #1-clique path decomposition}
\newcommand{\pw}{{\sf pw}}
\newcommand{\cell}{{\sf cell}}
\newcommand{\tw}{{\sf tw}}
\newcommand{\PP}{{\mathcal P}}
\newcommand{\TT}{{\mathcal T}}
\newcommand{\QQ}{{\cal Q}}
\newcommand{\WW}{{\cal W}}
\newcommand{\YY}{{\cal Y}}
\newcommand{\ZZ}{{\cal Z}}
\newcommand{\RR}{{\cal R}}
\newcommand{\CC}{{\cal C}}
\newcommand{\Yes}{{\sc Yes}}
\newcommand{\No}{{\sc No}}
\newcommand{\probSI}{\textsc{SI}\xspace}
\newcommand{\probSIlong}{\textsc{Subgraph Isomorphism}\xspace}
\newcommand{\gcgraphs}{clique-grid graphs\xspace}

\newcommand{\gcgraph}{clique-grid graph\xspace}
\newcommand{\cO}{\mathcal{O}}
\newcommand{\OO}{\mathcal{O}}

\newcommand{\probKPath}{\textsc{Longest Path}\xspace}
\newcommand{\probKCycle}{\textsc{Longest Cycle}\xspace}
\newcommand{\probKexactCycle}{\textsc{Exact $k$-Cycle}\xspace}
\newcommand{\probFVS}{\textsc{Feedback Vertex Set}\xspace}

\newcommand{\maxforestlong}{{\sc Max Induced Forest}}
\newcommand{\maxforest}{{\sc MIF}}
\newcommand{\probCycPacking}{\textsc{Cycle Packing}\xspace}

\newcommand{\defparproblem}[4]{
\vspace{1mm}
\noindent\fbox{
  \begin{minipage}{0.96\textwidth}
  \begin{tabular*}{\textwidth}{@{\extracolsep{\fill}}lr} #1 & {\bf{Parameter:}} #3 \\\end{tabular*}
  {\bf{Input:}} #2  \\
  {\bf{Question:}} #4
  \end{minipage}
  }
\vspace{1mm}
}

\title{Finding, Hitting and Packing Cycles in Subexponential Time on Unit Disk Graphs\footnote{Supported by Pareto-Optimal Parameterized Algorithms, ERC Starting Grant 715744,  Parameterized Approximation, ERC Starting Grant 306992, 
and  Rigorous Theory of Preprocessing, ERC Advanced Investigator Grant 267959.}}
\urldef{\mailUIB}\path|{fomin,daniello,fahadn.panolan,meirav.zehavi}@ii.uib.no|
\urldef{\mailIMSC}\path|saket@imsc.res.in|

\author[1]{Fedor V. Fomin}
\author[1]{Daniel Lokshtanov}
\author[1]{Fahad Panolan}
\author[1,2]{Saket Saurabh}
\author[1]{Meirav Zehavi}
\affil[1]{Department of Informatics, University of Bergen, Norway\\
\mailUIB}

\affil[2]{The Institute of Mathematical Sciences, HBNI, Chennai, India\\
\mailIMSC
}

\begin{document}
\date{}
\maketitle

%!TEX root = main_unit_disk.tex
\begin{abstract}
We give algorithms with running time $2^{\cO({\sqrt{k}\log{k}})} \cdot n^{\cO(1)}$ for the following  problems. Given an  $n$-vertex unit disk graph $G$ and an integer $k$, decide whether $G$ contains 
\begin{itemize}
\setlength{\itemsep}{-1pt}
\item  a path on exactly/at least $k$ vertices,
\item  a cycle on exactly $k$ vertices,
\item  a cycle on at least $k$ vertices, 
\item  a feedback vertex set of size at most $k$, and 
\item  a set of $k$ pairwise vertex-disjoint cycles.  
\end{itemize}
For the first three problems, no subexponential time parameterized algorithms were previously known. 
For the remaining two problems, our algorithms significantly outperform the previously best known 
parameterized algorithms that run in time $2^{\cO(k^{0.75}\log{k})} \cdot n^{\cO(1)}$.  
Our algorithms are based on a new kind of tree decompositions of unit disk graphs where the 
separators can have size up to $k^{\cO(1)}$ and there exists a solution that crosses every separator at most 
$\cO(\sqrt{k})$ times. The running times of our algorithms are optimal up to the $\log{k}$ factor in the exponent, assuming the Exponential Time Hypothesis. 
\end{abstract}
%These are the first parameterized subexponential time algorithms for the first three problems.
%For the remaining two problems the previously best known parameterized algorithms had running time 
%$2^{\cO(k^{0.75}\log{k})} \cdot n^{\cO(1)}$
%, our algorithms are the first to solve them in subexponential time algorithms
 %and significantly faster one for the remaining two  on unit disk graphs. 
% Our algorithms are based on dynamic programming over graphs of unbounded treewidth merged together with geometric insights.
%Our algorithms are based on novel insights on tree-decompsoitions of unit disk graphs. 
%!TEX root = main_unit_disk.tex
\section{Introduction}\label{sec:intro}
Unit disk graphs are the intersection graphs of unit circles in the plane. That is, given $n$-unit circles in the plane, we have  a graph $G$ where each vertex corresponds to a circle such that there is an edge  between two vertices when the corresponding circles intersect. Unit disk graphs form one of the most well studied graph classes in computational geometry because of their use in modelling optimal facility location~\cite{wang1988study} and broadcast networks such as wireless, ad-hoc and sensor networks \cite{hale1980frequency,kammerlander1984c,yeh1984outage}. 
These applications have led to an extensive study of NP-complete problems on unit disk graphs in the realms of computational complexity and approximation algorithms. We refer the reader 
to~\cite{ClarkCJ90,DumitrescuP11,HuntMRRRS98} and the citations therein for these studies. However, these problems remain hitherto unexplored in the light of parameterized complexity with exceptions that are few and far between~\cite{alber2002geometric,Chan03,FominLS12,Jansen10,SmithW98}.

In this paper we  consider the following basic problems about finding, hitting and packing cycles on unit disk graphs from the viewpoint of parameterized algorithms. For a given graph $G$ and integer $k$, 
\begin{itemize}
\item \probKexactCycle asks  whether $G$ contains a cycle on exactly $k$ vertices,
\item \probKCycle asks whether $G$ contains  a cycle on at least $k$ vertices, 
\item  \probFVS asks whether $G$ contains a vertex set $S$ of size $k$ such that the graph $G\setminus S$ is acyclic,  and 
\item  \probCycPacking  asks whether $G$ contains   a set of $k$ pairwise vertex-disjoint cycles.  
%\item \probKPath asks whether $G$ contains a path on at least $k$ vertices,
\end{itemize}
Along the way, we also study \probKPath (decide whether $G$ contains a path on exactly/at least $k$ vertices) and \probSIlong (\probSI). 
In \probSI, given {\em connected} graphs $G$ and $H$ on $n$ and $k$ vertices, respectively, the goal  is to decide whether there exists a subgraph in $G$ that is isomorphic to $H$.  T hroughout the paper  we assume that a unit disk graph is given by a set of $n$ points in the Euclidean plane and there is a  graph where vertices correspond to these points and there is an edge between two vertices if and only if the distance between the two points is at most $2$. 

%\item \probKPath asks whether $G$ contains a path on at least $k$ vertices,
%In the 
%\probKPath problem we are given an   $n$-vertex graph $G$ and integer $k$. The question is to decide whether $G$ contains a path of length %$k$. 

In parameterized complexity each of these problems serves as a testbed for development of  fundamental algorithmic techniques such as color-coding~\cite{AlonYZ}, the polynomial method   \cite{Koutis08,KoutisW16,Williams09,BjHuKK10},  matroid based techniques  \cite{FominLPS16} for \probKPath and \probKCycle, and kernelization techniques  for   \probFVS~\cite{Thomasse10}. 
We refer to  % \cite{FominK13},  \cite{KoutisW16} and 
\cite{cygan2015parameterized}
  for an extensive overview of the literature on parameterized algorithms for  these problems.  For example, the fastest known algorithms  solving \probKPath are the 
$1.66^k \cdot n^{\cO(1)}$ time 
  randomized algorithm of   Bj{\"o}rklund et al. 
  \cite{BjHuKK10},   and the $2.597^k \cdot n^{\cO(1)}$ time deterministic algorithm of Zehavi~\cite{Zehavi14}. Moreover, unless the Exponential Time Hypothesis (ETH) of  Impagliazzo, Paturi and Zane 
\cite{ImpagliazzoPZ01} fails, none of the problems above can be solved in time  $2^{o(k)} \cdot n^{\cO(1)}$~\cite{ImpagliazzoPZ01}.  

 While all these problems remain NP-complete on planar graphs,  substantially faster---\emph{subexponential}---parameterized algorithms  are known on planar graphs. In particular, by combining  the bidimensionality theory  of Demaine et al. \cite{DemaineFHT05jacm} with efficient   algorithms on graphs of bounded treewidth \cite{DornPBF10}, 
 \probKPath, \probKCycle,  \probFVS{} and  \probCycPacking are solvable in time $2^{\cO(\sqrt{k})}n^{\cO(1)}$ on planar graphs.
 The parameterized subexponential ``tractability'' of such problems can be extended to  graphs excluding some fixed graph as a
  minor~\cite{Demaine:2008mi}. The bidimensionality arguments cannot be applied to    \probKexactCycle and this was one of the motivations for developing the   new pattern-covering technique, which is used to   give a randomized algorithm for   \probKexactCycle running in time  $2^{\cO(\sqrt{k} \log^2 k)} n^{\cO(1)}$ on planar and apex-minor-free graphs ~\cite{FominLMPPS16}. The bidimensionality theory was 
 also used to design (efficient) polynomial time approximation scheme ((E)PTAS)~\cite{DemaineHaj05,FominLRS11} and 
 polynomial kernelization~\cite{F.V.Fomin:2010oq} on planar graphs.

It would be interesting to find generic properties  of problems,  similar to the theory of bidimensionality for planar-graph problems, that could guarantee the existence of  subexponential parameterized algorithms or (E)PTAS on geometric classes of graphs, such as unit disk graphs. The theory of (E)PTAS on geometric classes of graphs is extremely well developed and several methods have been devised for this purpose. This includes methods 
such as shifting techniques, geometric sampling and 
bidimensionality theory~\cite{HuntMRRRS98,Har-PeledQ15,Har-PeledL12,HochbaumM85,ClarksonV07,MustafaRR14,FominLS12}.
%A huge progress in this regard is achieved by designing PTAS for a large class of problems on unit disk graphs and map graphs  using the theory of bidimensionality~\cite{HuntMRRRS98,MustafaRR14,FominLS12}. 
However, 
we are still very far from a satisfactory understanding of the ``subexponential'' phenomena for problems on geometric graphs. We know that some problems such as  \textsc{Independent Set} and \textsc{Dominating Set}, which are solvable in time  $2^{\cO(\sqrt{k})}n^{\cO(1)}$  on planar graphs,  are W[1]-hard on unit disk graphs and thus the existence of an algorithm of running time   $f(k) \cdot n^{\cO(1)}$  is highly unlikely for any function $f$~\cite{Marx05}.
The existence of a vertex-linear kernel  for some problems on unit disk graphs such as 
 \textsc{Vertex Cover}~\cite{ChenKJ01} or  \textsc{Connected Vertex Cover}~\cite{Jansen10}  combined with an appropriate separation theorem~\cite{alber2002geometric,Chan03,SmithW98} yields a parameterized subexponential algorithm. A subset of the authors of this paper  used a different approach based on bidimensionality theory to obtain subexponential algorithms of running time $2^{\cO(k^{0.75}\log{k})} \cdot n^{\cO(1)}$ on unit disk graphs for 
\probFVS and  \probCycPacking   in~\cite{FominLS12}. No parameterized subexponential algorithms on unit disk graphs  for 
\probKPath, \probKCycle, and   \probKexactCycle were known prior to our work.

%Alber  and Fiala \cite{alber2002geometric} $n^{\cO(\sqrt{k})}$ algorithm for deciding whether an $n$-vertex unit disk graph contains an independent set of size $k$. 
%On the other hand,  as it was shown by Marx in \cite{Marx05}, the \textsc{Independent Set} problem is W[1]-hard on unit disk graphs and thus the existence of an algorithm of running time   $f(k) \cdot n^{\cO(1)}$ for some function $f$, is highly unlikely.
 
%
%exploiting the existence of a vertex-linear kernel and then use 
%a separator theorems for geometric graphs which are similar to Lipton-Tarjan separator theorem for planar graphs \cite{LiptonT79}, see e.g. \cite{Chan03,SmithW98}

 % 
%  All these subexponential algorithms are based on the bidimensionality theory of Demaine et al. \cite{DemaineFHT05jacm}.
%Very recently, Fomin et al.  \cite{FominLMPPS16} used a new pattern-covering technique to obtain  a subexponential algorithm for \probKPath on planar \emph{directed}  graphs. However, none of these techniques work on geometric classes of graphs and prior to our work, no parameterized algorithm faster than algorithms on general  graphs, was known on any reasonable class of geometric graphs. 
%In this paper, we establish the first subexponential parameterized algorithm for \probKPath on unit disk graphs, the  class of intersection graphs of unit disks in the plane. Our algorithm can be easily modified to work on  with other ``similar'' geometric graphs, like intersection graphs of unit squares in the plane. 

\noindent 
{\bf Our Results.}  We design subexponential parameterized algorithms, with running time $2^{\cO({\sqrt{k}\log{k}})} \cdot n^{\cO(1)}$, for 
 \probKexactCycle,  \probKCycle, \probKPath, \probFVS and  \probCycPacking on unit disk graphs and unit square graphs. It is also possible to show by known NP-hardness reductions for problems on unit disk graphs   \cite{ClarkCJ90}
  that   an algorithm of running time   $2^{o({\sqrt{k}})} \cdot n^{\cO(1)}$ for any of our problems on unit disk graphs would imply that ETH fails. Hence our algorithms are asymptotically almost tight. Along the way we also design Turing kernels (in fact, many to one)  for  \probKexactCycle,  \probKCycle, \probKPath and \probSI.  That is, we give a polynomial time algorithm that given an instance of  \probKexactCycle or \probKCycle or \probKPath or \probSI, produces polynomially many reduced instances of size polynomial in $k$ such that the input instance is a \Yes-instance if and only if one of the reduced instances is. As a byproduct of this we obtain a 
  $2^{\cO(k \log k)} \cdot n^{\cO(1)}$ time algorithm for \probSI when $G$ is a unit disk graph and $H$ is an arbitrary connected graph.  It is noteworthy to remark that a simple disjoint union trick implies that  \probKexactCycle,  \probKCycle, \probKPath, and  \probSI do not admit a polynomial kernel on unit disk  graphs~\cite{BodlaenderDFH09}. Finally, we remark that we do not use Turing kernels to design our subexponential time algorithms except for   \probKexactCycle. The subexponential time parameterized algorithm for   \probKexactCycle also uses a ``double layering''  of Baker's technique~\cite{Baker94}. 

All our subexponential time algorithms have the following theme in common.
  If  an input $n$-vertex unit disk graph $G$ contains a clique of size ${\sf poly}(k)$ (such a clique can be found in polynomial time), then we have a trivial \Yes-instance or \No-instance, depending on the problem. Otherwise, we show that the unit disk graph $G$ in a \Yes-instance of the problem admits, sometimes after a polynomial time preprocessing, a specific type of $(\omega, \Delta, \tau)$-decomposition, where the meaning of $\omega$, $\Delta$ and $\tau$ is as follows. The vertex set of $G$ is partitioned into cliques $C_1, \dots , C_d$, each   of size at most $\omega =k^{\cO(1)}$.  We also require that after contracting each of the cliques $C_i$ to a single vertex, the maximum vertex degree $\Delta$  of the obtained graph $\tilde{G}$ is  $\cO(1)$, while the treewidth $\tau$  of  $\tilde{G}$ is  
  $ \cO(\sqrt{k})$. Moreover, the corresponding  tree decomposition  of  $\tilde{G}$  can be constructed  efficiently. We use  the tree decomposition of  $\tilde{G}$ to construct a tree decomposition of $G$
by ``uncontracting'' each of the contracted cliques $C_i$. While the width of the obtained tree decomposition of $G$ can be of order   $\omega \cdot \tau =    k^{\cO(1)}$, we show that  each of our parameterized problems can be solved in time 
  $f(\Delta)\cdot \omega^{f(\Delta)\cdot \tau}$.  Here we use dynamic programming over the constructed tree decomposition of $G$, however 
  there is a twist from the usual way of designing such algorithms. 
  This part of the algorithm is problem-specific---in order to obtain the claimed running time, we have to establish a very specific property for each of the problems.  Roughly speaking, the desired property   of a problem   is that it always admits  an optimal solution  such that 
   for every pair of  adjacent 
  bags $X, Y$ of the tree decomposition of $G$,  
  the number of edges of  this solution ``crossing'' a  cut between $X$ and $Y$  is $\cO(\sqrt{k})$. We remark that the above decomposition is {\em only} given in the introduction to present our ideas for all the algorithms in a unified way.

\section{Preliminaries}\label{sec:prelim}
For a positive integer $t$, we use $[t]$ as a shorthand for $\{1,2,\ldots,t\}$. Given a function $f: A\rightarrow B$ and a subset $A'\subseteq A$, let $f|_{A'}$ denote the restriction of the function $f$ to the domain $A'$. For a  function $f: A\rightarrow B$ and $B'\subseteq B$, 
$f^{-1}(B')$ denote the set $\{a\in A~:~f(a)\in B'\}$. For $t,t'\in {\mathbb N}$, a set $[t]\times [t']$, $i \in [t]$ and $j\in [t']$  we use 
$(*,j)$ and $(i,*)$ to denote the sets $\{(i',j)~:~i'\in [t]\}$ and $\{(i,j')~:~j'\in [t']\}$, respectively.  For a set $U$, we use $2^U$ 
to denote the power set of $U$. 

%\todo[inline]{define $P=P_1 P_2$ if $P_1$ ends and $P_2$ starts at same vertex and reverse of a path, define K[X]}

\subparagraph*{Graph Theory.} We use standard notation and terminology from the book of Diestel~\cite{NewDiestel} for graph-related terms which are not explicitly defined here. Given a graph $G$, $V(G)$ and $E(G)$ denote its vertex-set and edge-set, respectively. When the graph $G$ is clear from context, we denote $n=|V(G)|$ and $m=|E(G)|$. Given $U\subseteq V(G)$, we let $G[U]$ denote the subgraph of $G$ induced by $U$, and we let $G\setminus U$ denote the graph $G[V(G)\setminus U]$. 
For an edge subset $E$, we use $V(E)$ to denote the set of endpoints of edges in $E$ and $G[E]$ to denote 
the graph $(V(E),E)$. For $X,Y\subseteq V(G)$, we use $E(X)$ and $E(X,Y)$ to denote the 
edge sets $\{\{u,v\}\in E(G)~:~u,v\in X\}$ and $\{\{u,v\}\in E(G)~:~u\in X, v\in Y\}$, respectively. 
Moreover, we let $N(U)$ denote the open neighborhood of $G$. In case $U=\{v\}$, we denote $N(v)=N(U)$. 
Given an edge $e=\{u,v\}\in E(G)$, we use $G/e$ to denote the graph obtained from $G$ by contracting the edge $e$. In other words, $G/e$ denotes the graph on the vertex-set $(V(G)\setminus\{u,v\})\cup\{x_{\{u,v\}}\}$, where $x_{\{u,v\}}$ is a new vertex, and the edge-set $E(G)=E(G[V(G)\setminus\{u,v\}])\cup\{\{x_{\{u,v\}},w\}~|~w\in N(\{u,v\})\}$. 
A graph $H$ is called a {\em minor} of $G$, if $H$ can be obtained from $G$ by a sequence of edge deletion, 
edge contraction and vertex deletion. 
In a graph $G$, a sequence of vertices $[u_1u_2\ldots u_{\ell}]$ is called a path in 
$G$, if  for any  $i,j\in [\ell]$, $i\neq j$, $u_i\neq u_j$ and $\{u_r,u_{r+1}\}\in E(G)$ for all $r\in [\ell-1]$.  
We also call the path $P=[u_1u_2\ldots u_{\ell}]$ as $u_1$-$u_{\ell}$ path. 
The internal vertices of a path $P=[u_1u_2\ldots u_{\ell}]$ are $\{u_2,u_3,\ldots,u_{\ell-1}\}$. 
For a path $P=[u_1u_2\ldots u_{\ell}]$, we use 
$\overleftarrow{P}$ to denote the path $[u_{\ell}u_{\ell-1}\ldots u_{1}]$. 
For any two paths $P_1=[u_1\ldots u_{i}]$ and $P_2=[u_i\ldots u_{\ell}]$, we use 
$P_1P_2$ to denote the path  $[u_1u_2\ldots u_{\ell}]$. 
A sequence of vertices $[u_1u_2\ldots u_{\ell}]$ is called a cycle in 
$G$, if  $u_1=u_{\ell}$, $[u_1u_2\ldots u_{\ell-1}]$ is a path and $\{u_{\ell-1},u_{\ell}\}\in E(G)$. 
For a path or a cycle $Q$, we use $V(Q)$ to denote the set of vertices in $Q$. 
Given $k\in\mathbb{N}$, we let $K_k$ denote  the compete graph on $k$ vertices. 
For a set $X$, we use $K[X]$ to denote the complete graph on $X$. 
Given $a,b\in\mathbb{N}$, an $a\times b$ grid is a graph on $a\cdot b$ vertices, $v_{i,j}$ for $(i,j)\in[a]\times[b]$, such that for all $i\in[a-1]$ and $j\in[b]$, it holds that $v_{i,j}$ and $v_{i+1,j}$ are neighbors, and for all $i\in[a]$ and $j\in[b-1]$, it holds that $v_{i,j}$ and $v_{i,j+1}$ are neighbors. For ease of presentation, for any function $f: D\rightarrow [a]\times [b]$, $i\in [a]$ and $j\in [b]$, we use 
$f^{-1}(i,j)$, $f^{-1}(*,j)$, and $f^{-1}(i,*)$ to denote the sets $f^{-1}((i,j))$, $f^{-1}((*,j))$, and $f^{-1}((i,*))$, respectively.   

A path decomposition is defined as follows.
\begin{definition}\label{def:pathDecomp}
A {\em path decomposition} of a graph $G$ is a sequence $\PP=(X_1,X_2,\ldots,X_{\ell})$, where each $X_i\subseteq V(G)$ is called a {\em bag}, that satisfies the following conditions. 
\begin{itemize}
\item $\bigcup_{i\in [\ell]}{X_i}=V(G)$.
\item For every edge $\{u,v\}\in E(G)$ there exists $i\in [\ell]$ such that  $\{u,v\}\subseteq X_{i}$.
\item For every vertex $v\in V(G)$, if $v\in X_i\cap X_j$ for some $i\leq j$, then $v\in X_r$ 
for all $r\in \{i,\ldots,j\}$.  
\end{itemize}
The {\em width} of $\PP$ is $\max_{i\in [\ell]} |X_i|-1$.
\end{definition} 

The {\em pathwidth} of $G$ is the minimum width of a path decomposition of $G$, and it is denoted by $\pw(G)$. A tree decomposition is a structure more general than a path decomposition, which is defined as follows.

\begin{definition}\label{def:treeDecomp}
A {\em tree decomposition} of a graph $G$ is a pair $\TT=(T,\beta)$, where $T$ is a tree and $\beta$ is a function from $V(T)$ to $2^{V(G)}$, that satisfies the following conditions.
\begin{itemize}
\item $\bigcup_{x\in V(T)}\beta(x)=V(G)$.
\item For every edge $\{u,v\}\in E(G)$ there exists $x\in V(T)$ such that  $\{u,v\}\subseteq \beta(x)$.
\item For every vertex $v\in V(G)$, if $v\in \beta(x)\cap \beta(y)$ for some $x,y\in V(T)$, then $v\in \beta(z)$ 
for all $z$ on the unique path between $x$ and $y$ in $T$.  
\end{itemize}
The {\em width} of $\TT$ is $\max_{x\in V(T)} |\beta(x)|-1$. Each $\beta(x)$ is called a {\em bag}. Moreover, we let $\gamma(x)$ denote the union of the bags of $x$ and its descendants.
\end{definition}

In other words, a path decomposition is a tree decomposition where $T$ is a path, but it will be convenient for us to think of a path decomposition as a sequence using the syntax in Definition \ref{def:pathDecomp}. The {\em treewidth} of $G$ is the minimum width of a tree decomposition of $G$, and it is denoted by $\tw(G)$.

\begin{proposition}[\cite{BodlaenderDDFLP16}]\label{prop:treewidth}
Given a graph $G$ and an integer $k$, in time $2^{\OO(k)} \cdot n$, we can either decide that $\tw(G)>k$ or 
output a tree decomposition of $G$ of width $5k$.
\end{proposition}

A {\em nice tree decomposition} is a tree decomposition of a form that simplifies the design of dynamic programming (DP) algorithms. Formally,

\begin{definition}
A tree decomposition $\TT=(T,\beta)$ of a graph $G$ is {\em nice} if for the root $r$ of $T$, it holds that $\beta(r)=\emptyset$, and each node $v\in V(T)$ is of one of the following types.
\begin{itemize}
\item {\bf Leaf}: $v$ is a leaf in $T$ and $\beta(v)=\emptyset$.
\item {\bf Forget}: $v$ has exactly one child $u$, and there exists a vertex $w\in\beta(u)$ such that $\beta(v)=\beta(u)\setminus\{w\}$.
\item {\bf Introduce}: $v$ has exactly one child $u$, and there exists a vertex $w\in\beta(v)$ such that $\beta(v)\setminus\{w\}=\beta(u)$.
\item {\bf Join}: $v$ has exactly two children, $u$ and $w$, and $\beta(v)=\beta(u)=\beta(w)$.
\end{itemize}
\end{definition}

\begin{proposition}[\cite{Bodlaender96}]\label{prop:nice}
Given a graph $G$ and a tree decomposition $\TT$ of $G$, a nice tree decomposition $\TT'$ of the same width as $\TT$ can be computed in linear time.
\end{proposition}

\subparagraph*{Geometric Graphs.} Given a set of geometric objects, $O$, we say that a graph $G$ {\em represents} $O$ if each vertex in $V(G)$ represents a distinct geometric object in $O$, and every geometric object in $O$ is represented by a distinct vertex in $V(G)$. In this case, we abuse notation and write $V(G)=O$. The {\em intersection graph of $O$} is a graph $G$ that represent $O$ and satisfies $E(G)=\{\{u,v\}: u,v\in O$, $u\cap v\neq\emptyset\}$.

Let $P=\{p_1=(x_1,y_1),p_2=(x_2,y_2),\ldots,p_n=(x_n,y_n)\}$ be a set of points in the Euclidean plane.
In the {\em unit disk graph model}, for every $i\in [n]$, we let $d_i$ denote the disk of radius 1 whose centre is $p_i$. Accordingly, we denote $D=\{d_1,d_2,\ldots,d_n\}$. Then, the {\em unit disk graph of} $D$ is the intersection graph of $D$. Alternatively, the unit disk graph of $D$ is the geometric graph of $G$ such that $E(G)=\{\{p_i=(x_i,y_i),p_j=(x_j,y_j)\}~|~p_i,p_j\in D, i\neq j, \sqrt{(x_i-x_j)^2+(y_i-y_j)^2}\leq 2\}$. In the {\em unit square graph model}, for every $i\in[n]$, we let $s_i$ denote the axis-parallel unit square whose centre is $p_i$. Accordingly, we denote $S=\{s_1,s_2,\ldots,s_n\}$. Then, the {\em unit square graph of} $S$ is the intersection graph of $S$. Alternatively, the unit square graph of $S$ is the geometric graph of $G$ such that  $E(G)=\{\{p_i=(x_i,y_i),p_j=(x_j,y_j)\}~|~p_i,p_j\in S, i\neq j, |x_i-x_j|\leq 1, |y_i-y_j|\leq 1\}$.

%!TEX root = main_unit_disk.tex

\section{Clique-Grid Graphs}

In this section, we introduce a family of ``grid-like'' graphs, called \gcgraphs, that is tailored to fit our techniques. Given a unit disk/square graph $G$, we extract the properties of $G$ that we would like to exploit, and show that they can be captured by an appropriate \gcgraph. Let us begin by giving the definition of a \gcgraph. Roughly speaking, a graph $G$ is a \gcgraph{} if each of its vertices can be embedded into a single cell of a grid (where multiple vertices can be embedded into the same cell), ensuring that the subgraph induced by each cell is a clique, and  that each cell can interact (via edges incident to its vertices) only with cells at ``distance'' at most 2. Formally,

\begin{definition}[\gcgraph]\label{def:GridClique}
A graph $G$ is a {\em \gcgraph{}} if there exists a function $f: V(G)\rightarrow [t]\times [t']$,  for some $t,t'\in {\mathbb N}$, such that the following conditions are satisfied.
\begin{enumerate}
\item\label{condition:GridClique1} For all $(i,j)\in [t]\times [t']$, it holds that $f^{-1}(i,j)$ is a clique.
\item\label{condition:GridClique2} For all $\{u,v\}\in E(G)$, it holds that if $f(u)=(i,j)$ and $f(v)=(i',j')$ then  
$$|i-i'|\leq 2\mbox{ and } |j-j'|\leq 2.$$ 
%$$f(v)\in \{(i',j')~|~|i-i'|\leq 2, |j-j'|\leq 2\},$$ where $f(u)=(i,j)$.
\end{enumerate}
Such a function $f$ is a {\em representation} of $G$. 
\end{definition}
We note that a notion similar to \gcgraph was also used by Ito and Kadoshita~\cite{ito2010tractability}. 
For the sake of clarity, we say that a pair $(i,j)\in [t]\times [t']$ is a {\em cell}. Moreover, whenever we discuss a clique-grid graph, we assume that we also have the representation. 
%As we suppose that an instance of \probSI on unit disk/square graphs is given together with a point set, we %suppose that an instance of \probSI on \gcgraphs{} is given together with a representation. 
Next, we show that a unit disk graph is a \gcgraph.

\begin{lemma}\label{lem:unitDisk}
Let $D$ be a set of points in the Euclidean plane, and let $G$ be the unit disk graph of $D$. Then, a representation $f$ of $G$ can be computed in polynomial time.
\end{lemma}

\begin{proof}
Denote $x_{\min}=\min\{x_i~|~p_i=(x_i,y_i)\in D\}$, $x_{\max}=\max\{x_i~|~p_i=(x_i,y_i)\in D\}$, $y_{\min}=\min\{y_i~|~p_i=(x_i,y_i)\in D\}$ and $y_{\max}=\max\{y_i~|~p_i=(x_i,y_i)\in D\}$. Accordingly, denote $\widehat{t}=\displaystyle{\frac{x_{\max}-x_{\min}}{\sqrt{2}}}$ and $\widehat{t}'=\displaystyle{\frac{y_{\max}-y_{\min}}{\sqrt{2}}}$. If $\widehat{t}=\lceil \widehat{t}\rceil$, then denote $t=\widehat{t}+1$, and otherwise denote $t=\widehat{t}$. Similarly, if $\widehat{t}'=\lceil \widehat{t}'\rceil$, then denote $t'=\widehat{t}'+1$, and otherwise denote $t'=\widehat{t}'$. Now, define $f: V(G)\rightarrow [t]\times [t']$ as follows. For all $p_i=(x_i,y_i)\in V(G)$, define $a_i=\lfloor \displaystyle{\frac{x_i-x_{\min}}{\sqrt{2}}}+1\rfloor$, $b_i=\lfloor \displaystyle{\frac{y_i-y_{\min}}{\sqrt{2}}}+1\rfloor$ and $f(p_i)=(a_i,b_i)$.

First, let us verify that Condition \ref{condition:GridClique1} in Definition \ref{def:GridClique} is satisfied. To this end, let $p_i=(x_i,y_i)$ and $p_j=(x_j,y_j)$ be two distinct vertices in $V(G)$ such that $f(p_i)=f(p_j)$. Then, $\lfloor \displaystyle{\frac{x_i-x_{\min}}{\sqrt{2}}}+1\rfloor=\lfloor \displaystyle{\frac{x_j-x_{\min}}{\sqrt{2}}}+1\rfloor$ and $\lfloor \displaystyle{\frac{y_i-y_{\min}}{\sqrt{2}}}+1\rfloor=\lfloor \displaystyle{\frac{y_j-y_{\min}}{\sqrt{2}}}+1\rfloor$. Thus, we have that $|x_i-x_j|<\sqrt{2}$ and $|y_i-y_j|<\sqrt{2}$. In particular, $\sqrt{(x_i-x_j)^2+(y_i-y_j)^2}< 2$, which implies that $(p_i,p_j)\in E(G)$.

Next, let us verify that Condition \ref{condition:GridClique2} in Definition \ref{def:GridClique} is satisfied. To this end, let $\{p_i=(x_i,y_i),p_j=(x_j,y_j)\}\in E(G)$. Recall that $f(p_i)$ and $f(p_j)$ are denoted by $(a_i,b_i)$ and $(a_j,b_j)$, respectively. Thus, to prove that $f(p_j)\in\{(a',b')~|~|a_i-a'|\leq 2, |b_i-b'|\leq 2\}$, it should be shown that $|a_i-a_j|\leq 2$ and $|b_i-b_j|\leq 2$. By substituting $a_i,a_j,b_i$ and $b_j$, it should be shown that
\begin{itemize}
\item $|\lfloor \displaystyle{\frac{x_i-x_{\min}}{\sqrt{2}}}\rfloor-\lfloor \displaystyle{\frac{x_j-x_{\min}}{\sqrt{2}}}\rfloor|\leq 2$, and
\item $|\lfloor \displaystyle{\frac{y_i-y_{\min}}{\sqrt{2}}}\rfloor-\lfloor \displaystyle{\frac{y_j-y_{\min}}{\sqrt{2}}}\rfloor|\leq 2$.
\end{itemize}

We focus on the proof of the first item, as the proof of the second item is symmetric.
Without loss of generality, suppose that $x_j\leq x_i$. Then, it remains to show that $\lfloor \displaystyle{\frac{x_i-x_{\min}}{\sqrt{2}}}\rfloor-\lfloor \displaystyle{\frac{x_j-x_{\min}}{\sqrt{2}}}\rfloor\leq 2$.
Since $G$ is the unit disk graph of $D$ and $\{p_i,p_j\}\in E(G)$, it holds that $\sqrt{(x_i-x_j)^2 + (y_i-y_j)^2}\leq 2$. In particular, $x_i-x_j\leq 2$. Denote $X=\displaystyle{\frac{x_j-x_{\min}}{\sqrt{2}}}$. Then, $\lfloor \displaystyle{\frac{x_i-x_{\min}}{\sqrt{2}}}\rfloor-\lfloor \displaystyle{\frac{x_j-x_{\min}}{\sqrt{2}}}\rfloor \leq \lfloor X+\sqrt{2}\rfloor-\lfloor X\rfloor\leq 2$.
\end{proof}

Similarly, we show the following.

\begin{lemma}\label{lem:unitSquare}
Let $S$ be a set of points in the Euclidean plane, and let $G$ be the unit square graph of $S$. Then, a representation $f$ of $G$ can be computed in polynomial time.
\end{lemma}

\begin{proof}
Denote $x_{\min}=\min\{x_i~|~p_i=(x_i,y_i)\in D\}$, $x_{\max}=\max\{x_i~|~p_i=(x_i,y_i)\in D\}$, $y_{\min}=\min\{y_i~|~p_i=(x_i,y_i)\in D\}$ and $y_{\max}=\max\{y_i~|~p_i=(x_i,y_i)\in D\}$. Accordingly, denote $\widehat{t}=x_{\max}-x_{\min}$ and $\widehat{t}'=y_{\max}-y_{\min}$. If $\widehat{t}=\lceil \widehat{t}\rceil$, then denote $t=\widehat{t}+1$, and otherwise denote $t=\widehat{t}$. Similarly, if $\widehat{t}'=\lceil \widehat{t}'\rceil$, then denote $t'=\widehat{t}'+1$, and otherwise denote $t'=\widehat{t}'$. Now, define $f: V(G)\rightarrow [t]\times [t']$ as follows. For all $p_i=(x_i,y_i)\in V(G)$, define $a_i=\lfloor x_i-x_{\min}+1\rfloor$, $b_i=\lfloor y_i-y_{\min}+1\rfloor$ and $f(p_i)=(a_i,b_i)$.

First, let us verify that Condition \ref{condition:GridClique1} in Definition \ref{def:GridClique} is satisfied. To this end, let $p_i=(x_i,y_i)$ and $p_j=(x_j,y_j)$ be two distinct vertices in $V(G)$ such that $f(p_i)=f(p_j)$. Then, $\lfloor x_i-x_{\min}+1\rfloor=\lfloor x_j-x_{\min}+1\rfloor$ and $\lfloor y_i-y_{\min}+1\rfloor=\lfloor y_j-y_{\min}+1\rfloor$. Thus, we have that $|x_i-x_j|<1$ and $|y_i-y_j|<1$, which implies that $(p_i,p_j)\in E(G)$.

Next, let us verify that Condition \ref{condition:GridClique2} in Definition \ref{def:GridClique} is satisfied. To this end, let $\{p_i=(x_i,y_i),p_j=(x_j,y_j)\}\in E(G)$. Recall that $f(p_i)$ and $f(p_j)$ are denoted by $(a_i,b_i)$ and $(a_j,b_j)$, respectively. Thus, to prove that $f(p_j)\in\{(a',b')~|~|a_i-a'|\leq 2, |b_i-b'|\leq 2\}$, it should be shown that $|a_i-a_j|\leq 2$ and $|b_i-b_j|\leq 2$. In fact, we will actually prove that $|a_i-a_j|\leq 1$ and $|b_i-b_j|\leq 1$.
By substituting $a_i,a_j,b_i$ and $b_j$, it is sufficient to show that
\begin{itemize}
\item $|\lfloor x_i-x_{\min}\rfloor-\lfloor x_j-x_{\min}\rfloor|\leq 1$, and
\item $|\lfloor y_i-y_{\min}\rfloor-\lfloor y_j-y_{\min}\rfloor|\leq 1$.
\end{itemize}

We focus on the proof of the first item, as the proof of the second item is symmetric.
Without loss of generality, suppose that $x_j\leq x_i$. Then, it remains to show that $\lfloor x_i-x_{\min}\rfloor-\lfloor x_j-x_{\min}\rfloor\leq 1$.
Since $G$ is the unit disk graph of $D$, $\{p_i,p_j\}\in E(G)$ and $x_j\leq x_i$, it holds that $x_i-x_j\leq 1$. Denote $X=x_j-x_{\min}$. Then, $\lfloor x_i-x_{\min}\rfloor-\lfloor x_j-x_{\min}\rfloor\leq \lfloor X+1\rfloor-\lfloor X\rfloor\leq 1$.
\end{proof}

Consequently, we have the following.

\begin{corollary}\label{cor:firstPhaseToGrid}
Let $(G,O,H,k)$ ($(G,O,k)$) be an instance of \probSI (resp.~\probKCycle) on unit disk/square graphs. Then, in polynomial time, one can output a representation $f$ such that $(G,f,H,k)$ (resp.~$(G,f,k)$) is an instance of \probSI (resp.~\probKCycle) on \gcgraphs{} that is equivalent to $(G,O,H,k)$ (resp.~$(G,O,k)$). 
\end{corollary}

We conclude this section by introducing the definition of an \NCTD{$\ell$}, which is useful for doing our dynamic programming algorithms. 

\begin{definition}\label{def:lcliqueTreeDecomp}
A tree decomposition $\TT=(T,\beta)$ of a \gcgraph{} $G$ with representation $f$ is a {\em \NCTDlong{$\ell$}}, or simply an {\em \NCTD{$\ell$}}, if for the root $r$ of $T$, it holds that $\beta(r)=\emptyset$, and for each node $v\in V(T)$, it holds that
\begin{itemize}
\item There exist at most $\ell$ cells, $(i_1,j_1),\ldots,(i_\ell,j_\ell)$, such that $\beta(v)=\bigcup_{t=1}^\ell f^{-1}(i_t,j_t)$, and
\item The node $v$ is of one of the following types.
	\begin{itemize}
	\item {\bf Leaf}: $v$ is a leaf in $T$ and $\beta(v)=\emptyset$.
	\item {\bf Forget}: $v$ has exactly one child $u$, and there exists a cell $(i,j)\in[t]\times[t']$ such that $f^{-1}(i,j)\subseteq\beta(u)$ and $\beta(v)=\beta(u)\setminus f^{-1}(i,j)$.
	\item {\bf Introduce}: $v$ has exactly one child $u$, and there exists a cell $(i,j)\in[t]\times[t']$ such that $f^{-1}(i,j)\subseteq\beta(v)$ and $\beta(v)\setminus f^{-1}(i,j)=\beta(u)\setminus f^{-1}(i,j)$.
	\item {\bf Join}: $v$ has exactly two children, $u$ and $w$, and $\beta(v)=\beta(u)=\beta(w)$.
	\end{itemize}
\end{itemize} 
\end{definition}

A {\em \NCPDlong{$\ell$}}, or simply an {\em \NCPD{$\ell$}}, is an \NCTD{$\ell$} where $T$ is a path. In this context, for convenience, we use the notation referring to a sequence presented in Section \ref{sec:prelim}.

%!TEX root = main_unit_disk.tex

\section{The Cell Graph of a Clique-Grid Graph}

In this section, we introduce two compact representations of \gcgraphs. By examining these representations, we are able to infer information on the structure of \gcgraphs{} that are also unit disk/square graphs.

\begin{definition}[backbone]\label{def:backbone}
Given a \gcgraph{} $G$ with representation $f: V(G)\rightarrow[t]\times[t']$, an induced subgraph $H$ of $G$ is a {\em backbone} for $(G,f)$ if for every two distinct cells $(i,j),(i',j')\in[t]\times[t']$ for which there exist $u\in f^{-1}(i,j)$ and $v\in f^{-1}(i',j')$ such that $\{u,v\}\in E(G)$, there also exist $u'\in f^{-1}(i,j)$ and $v'\in f^{-1}(i',j')$ such that $\{u',v'\}\in E(H)$. If no induced subgraph of $H$ is a backbone for $(G,f)$, then $H$ is a {\em minimal backbone} for $(G,f)$.
\end{definition}

First, we bound the maximum degree of a minimal backbone.

\begin{lemma}\label{lem:backboneDeg}
Let $G$ be a \gcgraph{} with representation $f$, and let $H$ be a minimal backbone for $(G,f)$. Then, for all $(i,j)\in[t]\times[t']$, it holds that $|f^{-1}(i,j) \cap V(H)|\leq 24$. Furthermore, the maximum degree of $H$ is at most $599$.
\end{lemma}

\begin{proof}
By Condition \ref{condition:GridClique2} in Definition \ref{def:GridClique}, we have that for all cells $(i,j)\in[t]\times[t']$, it holds that $f^{-1}(i,j)\cap V(H)\leq |\{(i',j')\in[t]\times[t']\setminus\{(i,j)\}~|~|i-i'|\leq 2,|j-j'|\leq 2\}|\leq 24$. Thus, for all $(i,j)\in[t]\times[t']$, the degree in $H$ of a vertex in $f^{-1}(i,j)\cap V(H)$ is bounded by 
\begin{eqnarray*}
|(\bigcup_{\substack{(i',j')\in[t]\times[t']\\ |i-i'|\leq 2 \\ |j-j'|\leq 2}}f^{-1}(i,j)\cap V(H))\setminus\{v\}| &\leq& 
|\{(i',j')\in[t]\times[t']~|~|i-i'|\leq2,|j-j'|\leq 2\}|\cdot 24-1 \\
&=& 25\cdot 24-1=599
\qedhere
\end{eqnarray*}
\end{proof}

Note that it is easy to compute a minimal backbone. The most naive computation simply initializes $H=G$; then, for every vertex $v\in V(G)$, it checks if the graph $H\setminus\{v\}$ has the same backbone as $H$, in which case it updates $H$ to $H\setminus\{v\}$. Thus, we have the following.

\begin{observation}\label{obs:computeBackbone}
Given a \gcgraph{} $G$ with representation $f$, a minimal backbone $H$ for $(G,f)$ can be computed in polynomial time.
\end{observation}

To analyze the treewidth of a backbone, we need the following.\footnote{The paper~\cite{FominLS12} does not consider unit square graphs, but the arguments it presents for unit disk graphs can be adapted to handle unit square graphs as well.}

\begin{proposition}[\cite{FominLS12}]\label{prop:gridUnitDisk}
Any unit disk/square graph with maximum degree $\Delta$ contains a $\displaystyle{\frac{\tw}{100\Delta^3}\times\frac{\tw}{100\Delta^3}}$ grid as a minor.
\end{proposition}

Thus, we have the following.

\begin{lemma}\label{lem:backboneTW}
Given a \gcgraph{} $G$ that is a unit disk/square graph, a representation $f$ of $G$ and an integer $\ell \in\mathbb{N}$, in time $2^{\OO(\ell)}\cdot n^{\OO(1)}$, one can either correctly conclude that $G$ contains a $\displaystyle{\frac{\ell}{100\cdot599^3}\times\frac{\ell}{100\cdot599^3}}$ grid as a minor, or obtain a minimal backbone $H$ for $(G,f)$ with a nice tree decomposition ${\cal T}$ of width at most $5k$.
\end{lemma}

\begin{proof}
By Lemma \ref{lem:backboneDeg},  Observation \ref{obs:computeBackbone} and Proposition \ref{prop:gridUnitDisk}, in polynomial time, one can compute a minimal backbone of $H$ such that $H$ either contains a $\displaystyle{\frac{\ell}{100\cdot599^3}\times\frac{\ell}{100\cdot599^3}}$ grid as a minor or has treewidth at most $\ell$. Since $H$ is a subgraph of $G$, it holds that if $H$ contains an $a\times b$ grid as a minor, then $G$ also contains an $a\times b$ grid as a minor. Thus, by Propositions \ref{prop:nice} and \ref{prop:treewidth}, we conclude the proof of the lemma.
\end{proof}
We use Lemma~\ref{lem:backboneTW} with $\ell=\OO(\sqrt{k})$.  
Next, we define a more compact representation of a \gcgraph.

\begin{definition}[cell graph]\label{def:cellGraph}
Given a \gcgraph{} $G$ with representation $f: V(G)\rightarrow[t]\times[t']$, the {\em cell graph} of $G$, denoted by $\cell(G)$, is the graph on the vertex-set $\{v_{i,j}: i\in[t],j\in[t'],f^{-1}(i,j)\neq\emptyset\}$ and edge-set $\{\{v_{i,j},v_{i',j'}\}: (i,j)\neq(i',j'), \exists u\in f^{-1}(i,j)\exists v\in f^{-1}(i',j')\ \mbox{ such that }\ \{u,v\}\in E(G)\}$. 
\end{definition}

By Definitions \ref{def:backbone} and \ref{def:cellGraph}, we directly conclude the following.

\begin{observation}\label{obs:cellGraphMinor}
For a \gcgraph{} $G$, a representation $f$ of $G$ and a backbone $H$ for $(G,f)$, it holds that $\cell(G)$ is a minor of $H$.
\end{observation}

Since for any graph $G$ and a minor $H$ of $G$, it holds that $\tw(H)\leq\tw(G)$, we have the following.

\begin{observation}\label{obs:twCell}
For a \gcgraph{} $G$, a representation $f$ of $G$ and a backbone $H$ for $(G,f)$, it holds that $\tw(\cell(G))\leq\tw(H)$.
\end{observation}

Overall, from Lemma \ref{lem:backboneTW} and Observation \ref{obs:twCell}, we directly have the following.

\begin{lemma}\label{lem:mainCell}
Given a \gcgraph{} $G$ that is a unit disk/square graph, a representation $f$ of $G$ and an integer $\ell\in\mathbb{N}$, in time $2^{\OO(\ell)}\cdot n^{\OO(1)}$, one can either correctly conclude that $G$ contains a $\displaystyle{\frac{\ell}{100\cdot599^3}\times\frac{\ell}{100\cdot599^3}}$ grid as a minor, or compute a nice tree decomposition of $\cell(G)$ of width at most $5\ell$.
\end{lemma}

Note that a nice tree decomposition of $\cell(G)$ of width $5\ell$ corresponds to a \NCTD{$5\ell$} of $G$. In other words, Lemma~\ref{lem:mainCell} implies the following.

\begin{corollary}\label{cor:mainNCTD}
Given a \gcgraph{} $G$ that is a unit disk/square graph, a representation $f$ of $G$ and an integer $\ell\in\mathbb{N}$, in time $2^{\OO(\ell)}\cdot n^{\OO(1)}$, one can either correctly conclude that $G$ contains a $\displaystyle{\frac{\ell}{100\cdot599^3}\times\frac{\ell}{100\cdot599^3}}$ grid as a minor, or compute a \NCTD{$5\ell$} of $G$.
\end{corollary}
%!TEX root = main_unit_disk.tex

\section{Turing Kernels}\label{sec:turingKer}

For the sake of uniformity, throughout this section, we denote an instance $(G,O,k)$ ($(G,f,k)$) of \probKCycle{} on unit disk/square graphs (resp.~\gcgraphs) also by $(G,O,H,k)$ (resp.~$(G,f,H,k)$) where $H$ is the empty graph. Our objective is to show that both {\sc SI} and \probKCycle on unit disk/square graphs admit a Turing kernel. More precisely, we prove the following.

\begin{theorem}\label{thm:turingKerIS}
Let $(G,O,H,k)$ be an instance of \probSI (\probKCycle) on unit disk/square graphs. Then, in polynomial time, one can output a set ${\cal I}$ of instances of \probSI (resp.~\probKCycle) on unit disk/square graphs such that $(G,O,H,k)$ is a \Yes-instance if and only if at least one instance in ${\cal I}$ is a \Yes-instance, and for all $(\widehat{G},\widehat{O},\widehat{H},\widehat{k})\in{\cal I}$, it holds that $|V(\widehat{G})|=\OO(k^3)$, $|E(\widehat{G})|=\OO(k^4)$, $\widehat{H}=H$ and $\widehat{k}=k$.
\end{theorem}

To prove Theorem \ref{thm:turingKerIS}, we first need two definitions.

\begin{definition}
Let $G$ be a \gcgraph\ with representation $f: V(G)\rightarrow[t]\times[t']$, $H'$ be a subgraph of $G$, and $\ell\in\mathbb{N}$. We say that $H'$ is {\em $\ell$-stretched} if there exist cells $(i,j),(i',j')\in[t]\times[t']$ such that the following conditions are satisfied.
\begin{itemize}
\item It holds that $|i-i'|\geq 2\ell$ or $|j-j'|\geq 2\ell$ (or both).
\item It holds that $V(H')\cap f^{-1}(i,j)\neq\emptyset$ and $V(H')\cap f^{-1}(i',j')\neq\emptyset$.
\end{itemize}
\end{definition}

\begin{definition}
Let $I=(G,f: V(G)\rightarrow [t]\times[t'],k)$ be an instance of \probKCycle on \gcgraphs. We say that $I$ is a {\em stretched instance} if $G$ has a cycle $C$ that is $\ell$-stretched for some $\ell\geq 2k$.
\end{definition}

We proceed by proving two claims concerning solutions to \probKCycle{} and \probSI{} on \gcgraphs.

\begin{lemma}\label{lem:patternDistantCells}
Let $I=(G,f: V(G)\rightarrow [t]\times[t'],H,k)$ be an instance of \probSI{} on \gcgraphs. Then, for any subgraph $H'$ of $G$ that is isomorphic to $H$, it holds that $H'$ is not $2k$-stretched.
\end{lemma}

\begin{proof}
Let $H'$ be a subgraph of $G$ that is isomorphic to $H$. Denote $i_{\min}=\min\{i\in [t]: (\bigcup_{j\in[t']}f^{-1}(i,j))\cap V(H')\neq\emptyset\}$, $i_{\max}=\max\{i\in [t]: (\bigcup_{j\in[t']}f^{-1}(i,j))\cap V(H')\neq\emptyset\}$, $j_{\min}=\min\{j\in [t']: (\bigcup_{i\in[t]}f^{-1}(i,j))\cap V(H')\neq\emptyset\}$ and $j_{\max}=\max\{j\in [t']: (\bigcup_{i\in[t]}f^{-1}(i,j))\cap V(H')\neq\emptyset\}$. To prove that $H'$ is not $2k$-stretched, we need to prove that $i_{\max}-i_{\min}<2k$ and $j_{\max}-j_{\min}<2k$. We only prove that $i_{\max}-i_{\min}<2k$, as the proof that $j_{\max}-j_{\min}<2k$ is symmetric.

Let $i_1<i_2<\cdots<i_\ell$ for the appropriate $\ell$ be the set of indices $i\in [t]$ such that $\bigcup_{j\in[t']}f^{-1}(i,j))\cap V(H')\neq\emptyset\}$. Note that $i_1=i_{\min}$ and $i_\ell=i_{\max}$. We claim that for all $r\in[\ell-1]$, it holds that $i_{r+1}-i_r\leq 2$. Suppose, by way of contradiction, that there exists $r\in[\ell-1]$ such that $i_{r+1}-i_r>2$.
Recall that $H$ is a connected graph, and therefore $H'$ is also a connected graph. 
Thus, there exists an edge $\{u,v\}\in E(H')\subseteq E(G)$ and indices $i\leq i_r$ and $i'\geq i_{r+1}$ such that $u\in \bigcup_{j\in[t']}f^{-1}(i,j))$ and $v\in \bigcup_{j\in[t']}f^{-1}(i',j))$. However, this contradicts the fact that $f$ is a representation of $G$.

Now, since for all $r\in[\ell-1]$, we proved that $i_{r+1}-i_r\leq 2$, we have that that there exist at least $\displaystyle{\frac{i_{\max}-i_{\min}}{2}+1}$ indices $i\in[t]$ such that $\bigcup_{j\in[t']}f^{-1}(i,j))\cap V(H')\neq\emptyset$. However, $|V(H)|\leq k$. Therefore $\displaystyle{\frac{i_{\max}-i_{\min}}{2}+1}\leq k$, which implies that $i_{\max}-i_{\min}< 2k$.
\end{proof}

\begin{lemma}\label{lem:cycDistantCells}
Let $I=(G,f: V(G)\rightarrow [t]\times[t'],k)$ be an instance of \probKCycle on \gcgraphs. Then, it can be determined in polynomial time whether $I$ is a stretched instance, in which case it is also a \Yes-instance.
\end{lemma}

\begin{proof}
It is well known that for any given graph and pair of vertices in this graph, one can determine (in polynomial time) whether the given graph has a cycle that contains both given vertices by checking whether there exists a flow of size $2$ between them (see, e.g., \cite{DBLP:confBjorklundHT}). Thus, by considering every pair $(u,v)$ of vertices in $V(G)$ such that $|i-i'|\geq 2k$ or $|j-j'|\geq 2k$ (or both) where $f(u)=(i,j)$ and $f(v)=(i',j')$, we can determine (in polynomial time) whether $I$ is a stretched instance.

Now, suppose that $I$ is a stretched instance. Then, $G$ has a cycle $C$ that is $\ell$-stretched for some $\ell\geq 2k$. Note that $I'=(G,f,C,|V(C)|)$ is a \Yes-instance of \probSI{} on \gcgraphs. Thus, by Lemma~\ref{lem:patternDistantCells}, it holds that $C$ is not $2|V(C)|$-stretched. Therefore, $\ell < 2|V(C)|$, and since $\ell\geq 2k$, we conclude that $k<|V(C)|$. Thus, $I$ is a \Yes-instance.
\end{proof}

Next, we prove a statement similar to the one of Theorem \ref{thm:turingKerIS}, but which concerns \gcgraphs. Our method is inspired by Baker's technique~\cite{Baker94}.

\begin{lemma}\label{lem:turingKernISGrid}
Let $I=(G,f: V(G)\rightarrow [t]\times[t'],H,k)$ be an instance of \probSI (\probKCycle) on \gcgraphs. Then, in polynomial time, one can output a set ${\cal I}$ of instances of \probSI (resp.~\probKCycle) on \gcgraphs{} such that $(G,f,H,k)$ is a \Yes-instance if and only if at least one instance in ${\cal I}$ is a \Yes-instance, and for all $(\widehat{G},\widehat{f}: V(\widehat{G})\rightarrow[\widehat{t}]\times [\widehat{t}'],\widehat{H},\widehat{k})\in{\cal I}$, it holds that $\widehat{G}$ is either an induced subgraph of $G$ or $K_{\widehat{k}}$, $\widehat{t},\widehat{t}'\leq 2k=\OO(\widehat{k})$, $|\widehat{f}^{-1}(i,j)|<\widehat{k}$ for any cell $(i,j)\in[t]\times[t']$, $\widehat{H}=H$ and $\widehat{k}=k$.
\end{lemma}

\begin{proof}
First, suppose that there exists a cell $(i,j)\in[t]\times[t']$ such that $|f^{-1}(i,j)|\geq k$, then by Definition~\ref{def:GridClique}, $G[f^{-1}(i,j)]$ is a clique on at least $k$ vertices. In particular, the pattern $H$ is a subgraph of $G[f^{-1}(i,j)]$, and therefore it is also a subgraph of $G$. Thus, in this case, we conclude the proof by setting $\cal I$ to be the set that contains only one instance, $(K_{k},\widehat{f}: V(K_k)\rightarrow [1]\times[1],H,k)$. From now on, suppose that for all cells $(i,j)\in[t]\times[t']$, it holds that $|f^{-1}(i,j)|<k$.

Second, in case the input instance $I$ is of \probKCycle, we use the computation given by Lemma \ref{lem:cycDistantCells} to determine whether $I$ is a stretched  instance. If the answer is positive, then by Lemma \ref{lem:cycDistantCells}, it holds that $I$ is a \Yes-instance. In this case, we  again conclude the proof by setting $\cal I$ to be the set that contains only one instance, $(K_{k},\widehat{f}: V(K_k)\rightarrow [1]\times[1],H,k)$. From now on, also suppose that if the input instance $I$ is of \probKCycle, then it is not stretched .

Now, our kernelization algorithm works as follows. For every $(p,q)\in [t]\times[t']$, it computes
$$\displaystyle{G_{p,q}=G[\bigcup_{\substack{p\leq i<\min\{p+2k,t+1\}\\ q\leq j<\min\{q+2k,t'+1\}}}f^{-1}(i,j)]}.$$
Accordingly, it computes $f_{p,q}: V(G_{p,q})\rightarrow [\min\{2k,t\}]\times[\min\{2k,t'\}]$ as follows. For every $v\in V(G_{p,q})$, compute $f_{p,q}(v)=(i-p+1,j-q+1)$ where $(i,j)=f(v)$. Note that for all $(i,j)\in [\min\{2k,t\}]\times[\min\{2k,t'\}]$, it holds that $f^{-1}_{p,q}(i,j)=f^{-1}(i+p-1,j+q-1)$. Thus, since $f$ is a representation of $G$, it holds that $f_{p,q}$ is a representation of $G_{p,q}$. Finally, our kernelization algorithm outputs ${\cal I}=\{I_{p,q}=(G_{p,q},f_{p,q},H,k): (p,q)\in [t]\times[t']\}$.

To conclude the proof, it remains to show that $(G,f,H,k)$ is a \Yes-instance if and only if at least one instance in ${\cal I}$ is a \Yes-instance. Since for all $(G_{p,q},f_{p,q},H,k)\in{\cal I}$, it holds that $G_{p,q}$ is an induced subgraph of $G$, we have that if $(G,f,H,k)$ is a \No-instance, then every instance in ${\cal I}$ is \No-instance as well. Next, suppose that $(G,f,H,k)$ is a \Yes-instance. Let us consider two cases.

\begin{itemize}
\item $(G,f,H,k)$ is an instance of \probSI. Then, let $H'$ be a subgraph of $G$ that is isomorhpic to $H$. Denote $i_{\min}=\min\{i\in [t]: (\bigcup_{j\in[t']}f^{-1}(i,j))\cap V(H')\neq\emptyset\}$, $i_{\max}=\max\{i\in [t]: (\bigcup_{j\in[t']}f^{-1}(i,j))\cap V(H')\neq\emptyset\}$, $j_{\min}=\min\{j\in [t']: (\bigcup_{i\in[t]}f^{-1}(i,j))\cap V(H')\neq\emptyset\}$ and $j_{\max}=\max\{j\in [t']: (\bigcup_{i\in[t]}f^{-1}(i,j))\cap V(H')\neq\emptyset\}$. By Lemma \ref{lem:patternDistantCells}, it holds that both $i_{\max}-i_{\min}<2k$ and $j_{\max}-j_{\min}<2k$. Therefore, $H'$ is a subgraph of $G_{i_{\min},j_{\min}}$, which implies that $I_{p,q}$ is a \Yes-instance.

\item $(G,f,H,k)$ is an instance of \probKCycle. Then, let $C$ be a subgraph of $G$ that is a cycle on at least $k$ vertices. Denote $i_{\min}=\min\{i\in [t]: (\bigcup_{j\in[t']}f^{-1}(i,j))\cap V(C)\neq\emptyset\}$, $i_{\max}=\max\{i\in [t]: (\bigcup_{j\in[t']}f^{-1}(i,j))\cap V(C)\neq\emptyset\}$, $j_{\min}=\min\{j\in [t']: (\bigcup_{i\in[t]}f^{-1}(i,j))\cap V(C)\neq\emptyset\}$ and $j_{\max}=\max\{j\in [t']: (\bigcup_{i\in[t]}f^{-1}(i,j))\cap V(C)\neq\emptyset\}$. Since $(G,f,H,k)$ is not stretched, it holds that both $i_{\max}-i_{\min}<2k$ and $j_{\max}-j_{\min}<2k$. Therefore, $C$ is a subgraph of $G_{i_{\min},j_{\min}}$, which implies that $I_{p,q}$ is a \Yes-instance.\qedhere
\end{itemize}
\end{proof}

Towards proving Theorem \ref{thm:turingKerIS}, we extract a claim that is reused in Sections \ref{sec:exactCyc} and \ref{sec:longCyc}.

\begin{corollary}\label{cor:secPhaseToGrid}
Let $(G,O,H,k)$ be an instance of \probSI (\probKCycle) on unit disk/square graphs. Then, in polynomial time, one can output a set ${\cal I}$ of instances of \probSI (resp.~\probKCycle) on \gcgraphs{} such that $(G,O,H,k)$ is a \Yes-instance if and only if at least one instance in ${\cal I}$ is a \Yes-instance, and for all $(\widehat{G},\widehat{f},\widehat{H},\widehat{k})\in{\cal I}$, it holds that $\widehat{G}$ is either an induced subgraph of $G$ or $K_{\widehat{k}}$, $\widehat{t},\widehat{t}'=\OO(\widehat{k})$, $|\widehat{f}^{-1}(i,j)|<\widehat{k}$ for any cell $(i,j)\in[\widehat{t}]\times[\widehat{t}']$, $\widehat{H}=H$ and $\widehat{k}=k$.
\end{corollary}

\begin{proof}
First, by Corollary \ref{cor:firstPhaseToGrid}, we obtain (in polynomial time) an instance $(G,f,H,k)$ of \probSI (\probKCycle) on \gcgraphs{} that is equivalent to $(G,O,H,k)$. Then, by Lemma~\ref{lem:turingKernISGrid} with $(G,f,H,k)$, we obtain (in polynomial time) the desired set $\cal I$.
\end{proof}

We are now ready to prove Theorem \ref{thm:turingKerIS}.

\subparagraph*{Proof of Theorem \ref{thm:turingKerIS}.}
First, by Corollary \ref{cor:secPhaseToGrid}, we obtain (in polynomial time) a set ${\cal I}$ of instances of \probSI (\probKCycle) on \gcgraphs{} such that $(G,O,H,k)$ is a \Yes-instance if and only if at least one instance in ${\cal I}$ is a \Yes-instance, and for all $(\widehat{G},\widehat{f},\widehat{H},\widehat{k})\in{\cal I}$, it holds that $\widehat{G}$ is either an induced subgraph of $G$ or $K_{\widehat{k}}$, $\widehat{t},\widehat{t}'=\OO(\widehat{k}^2)$, $|\widehat{f}^{-1}(i,j)|\leq \widehat{k}$ for a cell $(i,j)\in[\widehat{t}]\times[\widehat{t}']$, $\widehat{H}=H$ and $\widehat{k}=k$. Thus, to conclude our proof, it is sufficient to show that for any instance $(\widehat{G},\widehat{f},\widehat{H},\widehat{k})\in{\cal I}$, we can compute (in polynomial time) an equivalent instance $(G',O',H',k')$ of \probSI (\probKCycle) on unit disk/square graphs such that $|O'|=\OO(k'^3)$, $H'=H$ and $k'=k$. To this end, fix some instance $(\widehat{G},\widehat{f},\widehat{H},\widehat{k})\in{\cal I}$. Let us first handle the simple case where $\widehat{G}$ is equal to $K_{\widehat{k}}$. Here, we conclude the proof by defining $p'_i=(0,i/k)$ for $i\in[k]$, and then setting $O'=\{p'_i: i\in[k]\}$.

Now, suppose that $\widehat{G}$ is an induced subgraph of $G$. Then, we have that $V(\widehat{G})\subseteq O$, and the unit disk/square graph of $V(\widehat{G})$ is exactly $\widehat{G}$. Thus, we define $(G',O',H',k')=(\widehat{G},V(\widehat{G}),\widehat{H},\widehat{k})$. It remains to show that $|V(\widehat{G})|=\OO(k^3)$,$|E(\widehat{G})|=\OO(k^4)$.

By Definition \ref{def:GridClique}, it holds that
$$|V(\widehat{G})|=\displaystyle{\sum_{(i,j)\in[\widehat{t}]\times[\widehat{t}']}|f^{-1}(i,j)|}.$$
Thus, since $\widehat{t},\widehat{t}'=\OO(k)$ and $|\widehat{f}^{-1}(i,j)|\leq k$ for $(i,j)\in[\widehat{t}]\times[\widehat{t}']$, we have that $|V(\widehat{G})|=\OO(k^3)$.

Now, denote $X=\{((i,j),(i',j')): i,i'\in[t], j,j'\in[t'], |i-i'|\leq 2, |j-j'|\leq 2\}$. Since $\widehat{t},\widehat{t}'=\OO(k)$, we have that $|X|=\OO(k^2)$. By Definition \ref{def:GridClique}, it also holds that
$$|E(\widehat{G})|\leq \displaystyle{\sum_{((i,j),(i',j'))\in X}|f^{-1}(i,j)|\cdot|f^{-1}(i',j')|}.$$
Thus, since $|X|=\OO(k^2)$, and $|\widehat{f}^{-1}(i,j)|\leq k$ for $(i,j)\in[\widehat{t}]\times[\widehat{t}']$, we have that $|E(\widehat{G})|=\OO(k^4)$.\qed

%!TEX root = main_unit_disk.tex

\section{Exact $k$-Cycle}\label{sec:exactCyc}

In this section we prove the following theorem. 
\begin{theorem}
\label{thm:exactcycle}
\probKexactCycle{}  on unit disk/square graphs  
can be solved  in $2^{\OO(\sqrt k\log k)} \cdot n^{\OO(1)}$ time. 
%when $n$ is the number of vertices in the input graph. 
\end{theorem}

Towards proving Theorem~\ref{thm:exactcycle}, 
we design an algorithm which given a \gcgraph{} $G$ along with its representation 
$f: V(G)\rightarrow [t]\times[t']$ and an integer $k$ as input, runs in time $2^{\OO(\sqrt k\log k)}\vert V(G)\vert^{\OO(1)}$ and 
decides whether $G$ has a cycle of length $k$ or not. In Section~\ref{sec:turingKer} (Lemma~\ref{lem:turingKernISGrid}) we have seen that \probSI admits a polynomial sized Turing kernel on \gcgraph{s}. Hence to give an algorithm of  running time  
$2^{\OO(\sqrt k\log k)}\vert V(G)\vert^{\OO(1)}$, we can restrict to instances of size bounded 
by polynomial in $k$.
More precisely, because of Lemma~~\ref{lem:turingKernISGrid}, we can assume that  the input to our algorithm is $(G,f: V(G)\rightarrow [t]\times [t'],k)$ 
where $G$ is a \gcgraph{} with a representation $f$, $\vert f^{-1}{(i,j)}\vert < k$ for all $(i,j)\in [t]\times [t']$  
and $t,t'\leq 2k$. Without loss of generality we can assume that  $f$ is a function from $V(G)$ to $[2k]\times [2k]$, 
because $[t]\times [t'] \subseteq [2k]\times [2k]$. 

Given an instance  $(G,f: V(G)\rightarrow [2k]\times [2k],k)$, the algorithm applies a method inspired by Baker's technique~\cite{Baker94}  and obtains a family,  $\mathscr F$, of  $2^{\OO(\sqrt k \log k)}$ instances of \probKexactCycle. The family $\mathscr F$ has following properties. 
\begin{enumerate}
\item  In each of these instances the input graph is an induced subgraph of $G$ and has size $k^{\OO(1)}$. 
\item  The input $(G,f: V(G)\rightarrow [2k]\times [2k],k)$ is a \Yes-instance if and only if  there exists an instance 
$(H,f^*: V(H)\rightarrow [2k]\times [2k],k) \in \mathscr F$ which is a  \Yes-instance. 
% such that 
%if there is a $k$ length cycle in $G$, then at least one instance created has a 
%$k$ length cycle. 
\item  More over, for any instance $(H,f^*: V(H)\rightarrow [2k]\times [2k],k) \in \mathscr F$,  
$H$ has a \NCPDlong{$7\sqrt k$} (\NCPD{$7\sqrt k$}). 
%Let $\PP=(X_1,\ldots,X_q)$ be a  \NCPD{$7\sqrt k$} of $H$. 
\end{enumerate}
We will call the family $\mathscr F$ satisfying the above properties as  {\em good family}. 
%
%That is, there is a sequence $\PP=(X_1,\ldots,X_q)$ of subsets of $V(H)$ such that 
%\begin{itemize}
%\item $\bigcup_{i=1}^q X_i = V(H)$, 
%\item $X_1=X_r=\emptyset$, and 
%\item for every $\ell \in [q-1]$, there is a cell $(i,j)$ of $f$ such either $(i)$ 
%$f|_{V(H)}^{-1}(i,j)\cap X_{\ell}=\emptyset$ and  
%$X_{\ell+1}=X_{\ell} \cup f|_{V(H)}^{-1}(i,j)$ or $(ii)$ $f|_{V(H)}^{-1}(i,j)\subseteq X_{\ell}$ and  
%$ X_{\ell+1}=X_{\ell}\setminus f|_{V(H)}^{-1}(i,j)$. 
%\end{itemize}
%The bags of the form $X_{\ell+1}=X_{\ell} \cup f^{-1}(i,j)$ are called {\em introduce bags} and 
%bags of the form $ X_{\ell+1}=X_{\ell}\setminus f^{-1}(i,j)$ are called {\em forget bags}. 
%
Let $(H,f^*: V(H)\rightarrow [2k]\times [2k],k)$ be an instance of  $\mathscr F$. 
Let $\PP=(X_1,\ldots,X_q)$ be a  \NCPD{$7\sqrt k$} of $H$. 
We first prove that if there is a cycle of length $k$ in $H$, then there is a cycle $C$ of length $k$ 
in $H$ such that for any two distinct cells $(i,j)$ and $(i',j')$ of $f$, the number of edges with one end point in 
$(i,j)$ and other in $(i',j')$ is at most $4$. 
Let $C$ be such a cycle in $H$. 
Then using the property of $C$ 
%and the special path decomposition created we will get 
we get the following important property.   
\begin{framed}
\noindent 
For any $i\in [q]$, the number of edges of $V(C)$ with one end point in $X_i$ and other in 
$\bigcup_{i<j\leq q} X_j$ is upper bounded by $\OO(\sqrt k)$.      
\end{framed}

The above mentioned property allows us to design a dynamic programming (DP) algorithm 
over 
%the path decomposition 
\NCPD{$7\sqrt k$},  
$\PP$, for \probKexactCycle{} in time $2^{\OO(\sqrt k\log k)}$. 
Now we are ready to give formal details about the algorithm. As explained before,  
we assume that the number of vertices in the input graph is bounded by $k^{\OO(1)}$. 

\begin{lemma}
\label{lem:exactcyclepw}
Let $(G,f: V(G)\rightarrow [2k]\times [2k],k)$ be an instance of \probKexactCycle, where 
$G$ is a \gcgraph{} with representation $f$, 
$\vert f^{-1}{(i,j)}\vert < k$ for all $(i,j)\in [2k]\times [2k]$ 
and $\vert V(G)\vert= k^{\OO(1)}$.   Given $(G,f: V(G)\rightarrow [2k]\times [2k],k)$, there is an algorithm running  in time $2^{\OO(\sqrt k\log k)}$ that outputs a good family  $\mathscr F$. 
%Given $(G,f: V(G)\rightarrow [2k]\times [2k],k)$, there is an algorithm runs in time $2^{\OO(\sqrt k\log k)}$ and  output 
%$2^{\OO(\sqrt k\log k)}$ many instances of  \probKexactCycle, where the input graph of each instance is an induced subgraph of $G$.   
%More over, each input graph $H$ in the instances created,  
%has  
%a \NCPD{$7\sqrt k$} 
%and it can be computed in polynomial time. 
\end{lemma}

%Let $C$ be a $k$ length path in $G$ satisfying Lemma~\ref{lemm:path:shortinteraction}.  
\begin{proof}
Let $C$ be a $k$ length cycle in $G$. 
First we define a column of the  $2k\times 2k$ grid. For any $j\in [2k]$ 
the set of cells $\{(i,j)~:~i\in [2k]\}$ is called a column. There are $2k$ columns 
for the  $2k\times 2k$ grid. 
We partition $2k$ columns of the  $2k\times 2k$ grid with $k$ blocks 
of two consecutive columns and label them from the set of labels $[\sqrt k]$. 
That is, each pair of columns $2i-1$ and $2i$, where $i\in [k]$ is labelled 
with $i \mod \sqrt k$. In other words both column $2i-1$ and $2i$ are labelled 
with $i \mod \sqrt k$. 
Then by pigeon hole principle 
there is a label $\ell \in \{1,2,\ldots, \sqrt k\}$ such that the number of vertices from $V(C)$ which are in 
columns labelled $\ell$ is at most $\sqrt k$. As $\vert V(G)\vert\leq k^{\OO(1)}$,  
the number of vertices of $G$ in columns labelled $\ell$ is at most  $k^{\OO(1)}$. We guess 
the vertices of $V(C)$ which are in the columns labelled $\ell$. The number of potential 
guesses is bounded by $k^{\OO(\sqrt k)}$. Let $Y$ be the set of guessed vertices of $V(C)$ 
which are in the columns labelled by $\ell$. Notice that $\vert Y \vert \leq \sqrt k$.   Then 
we delete all the vertices in columns labelled $\ell$, except the vertices of $Y$. Let $S$ be the 
set of deleted vertices. By the property $2$ of \gcgraph, $G\setminus (S\cup Y)$ is a disjoint union 
of   \gcgraph{s} each of which  is represented by a function with at most $2\sqrt k$ columns. 
%See Fig.~\ref{} for an illustration\todo{add a figure}.  
That is, let 
$G_1=G[\bigcup_{j=1}^{2(\ell-1)}f^{-1}(*,j)]$, and 
$G_{i+1}=G[\bigcup_{j=i\cdot 2\ell+1}^{\min \{ i\cdot 2\ell+2\sqrt k,2k\}}f^{-1}(*,j)]$ for all $i\in \{1,\ldots, \lceil \sqrt k\rceil \}$. 
Notice that $G_i$ is \gcgraph{} with representation $f_i : V(G_i)\rightarrow [2k]\times [2\sqrt k]$ defined as, 
$f_i(u)=(r,j)$, when $f(u)=(r,(i-1)2\ell+j)$.  By the property $2$ of \gcgraph, $G\setminus (S\cup Y)=G_1\uplus\ldots \uplus G_{\lceil \sqrt k\rceil  +1}$. 

\begin{claim}
\label{lemm:pwfewcells}
$G\setminus S$ has a 
\NCPDlong{$7\sqrt k$} (\NCPD{$7\sqrt k$}). 
%path decomposition $\PP$ such each bag is a union of vertices assigned to at most $7 \sqrt k$ many 
%cells of $f|_{V(G)\setminus S}$. In other words, for each bag $X\in \PP$, there exists cells $p_1,\ldots,p_{\ell}$ of $f$  
%such that $X=\bigcup_{i\in [\ell]} f^{-1}(p_i) \setminus S$, where $\ell \leq 7 \sqrt k$ 
\end{claim}
\begin{proof}
First, for each  $i\in \{1,\ldots,\lceil \sqrt k\rceil  +1\}$, we define a path decomposition of $G_i$  
such that each bag is a union of at most $6\sqrt k$ many 
cells of $f_i$.   As $G\setminus(S\cup Y)=G_1\uplus\ldots \uplus G_{\lceil \sqrt k\rceil  -1}$, and $\vert Y \vert \leq \sqrt k$, 
by adding $Y$ to each bag of all path decompositions we can get a required 
\NCPDlong{$7\sqrt k$} for 
$G\setminus S$.

Now, for each $G_i$, we define a path decomposition $\PP_i=(X_{i,1}, X_{i,2},\ldots X_{i, 2k-2})$ where 
$X_{i,j}=f_i^{-1}(j,*)\cup f_i^{-1}(j+1,*)\cup f_i^{-1}(j+2,*)$. We claim that $\PP_i$ is indeed a path decomposition of 
$G_i$. Notice that  $\bigcup_{j=1}^{k-1} X_{i,j}=f_i^{-1}(*,*)=V(G_i)$. By property $2$ of \gcgraph, 
we have that for each edge $\{u,v\}\in E(G)$, there exists $j\in [2k-2]$ such that $\{u,v\}\in X_{i,j}$. For each $u\in V(G)$, $u$ 
is contained in at most three bags and these bags are consecutive in the sequence $(X_{i,1}, X_{i,2},\ldots X_{i, 2k-2})$. 
Hence $\PP_i$ is a path decomposition of $G_i$. Since $X_{i,j}=f_i^{-1}(j,*)\cup f_i^{-1}(j+1,*)\cup f_i^{-1}(j+2,*)$, number of columns 
in $f_i$ is at most $2\sqrt k$ and each cell of $f_i$ is a cell of $f$, each $X_{i,j}$ is a union of $6 \sqrt k$ many 
cells of $f$. 
%We can turn the path decomposition $\PP_i$ obtained to a \NCPD{$6\sqrt k$} 
%by an algorithm similar to the one mentioned in
% Proposition~\ref{prop:nice}.
%Hence $\PP_i$ is a \NCPDlong{$6\sqrt k$} of $G_i$. 
%
%
Since $G\setminus (S\cup Y)=G_1\uplus\ldots \uplus G_{\lceil \sqrt k\rceil  +1}$,   
the sequence  
$\PP'=(X_{1,1},\ldots X_{1, 2k-2},X_{2,1},\ldots X_{2, 2k-2}, \ldots, X_{\lceil \sqrt k\rceil  -1,1},\ldots X_{\lceil \sqrt k\rceil  -1, 2k-2})$ is a 
path decomposition of $G\setminus (S\cup Y)$. 
%\NCPDlong{$6\sqrt k$} of $G\setminus (S \cup Y)$. 
More over, the vertices of each bag is a union of vertices from at most $6\sqrt k$  cells of $f$. 
Also, since $\vert Y \vert \leq \sqrt k$, 
the sequence $\PP=(X_{1,1}\cup Y ,\ldots X_{1, k-2}\cup Y,\ldots X_{\lceil \sqrt k\rceil  -1, k-2}\cup Y)$ 
obtained by adding $Y$ to each bag of $\PP'$ we get a  path decomposition of $G\setminus S$. 
More over, the vertices of each bag in $\PP$ is a union of vertices from at most $7\sqrt k$  cells of $f$. 
%We can turn the path decomposition obtained to a \NCPD{$7\sqrt k$} similar to the one mentioned in Proposition~\ref{prop:nice}. 
We can turn the path decomposition $\PP$ to a \NCPD{$7\sqrt k$} 
by an algorithm similar to the one mentioned in
 Proposition~\ref{prop:nice}.
\end{proof}
Our algorithm will construct a family $\mathscr F$ as follows. For each $\ell \in \{1,\ldots, \lceil \sqrt k\rceil\}$ and 
for two subsets of vertices $S$ and $Y$  such 
that $S\cup Y$ is a set of vertices in the columns labelled $\ell$ and 
$\vert Y \vert \leq \sqrt k$, our algorithm will include an instance 
$(G\setminus S, f|_{V(G)\setminus S},k)$ in $\mathscr F$. The number of choices of $S$ and $Y$ is bounded 
by $2^{\OO(\sqrt k\log k)}$ and thus the size of $\mathscr F$ is bounded by $2^{\OO(\sqrt k\log k)}$. 

We claim that $\mathscr F$ is indeed a good family. Suppose there is a cycle $C$ of length $k$ in $G$. 
Then, by pigeon hole principle there is $\ell \in \{1,\ldots, \lceil \sqrt k\rceil\}$ such that at most $\sqrt{k}$ 
vertices from $V(C)$ are in the columns labelled by $\ell$. Let $S'$ be the set of vertices in the columns labelled by $\ell$. 
 Let  $Y=S'\cap V(C)$ and $S=S'\setminus Y$. Notice that $\vert Y \vert \leq \sqrt k$.  The instance $(G\setminus S, f|_{V(G)\setminus S},k)$ 
in  $\mathscr F$, is a \Yes\ instance.  
This completes the proof of the lemma.
\end{proof}

%
%\begin{figure}
%%\begin{subfigure}[b]{0.5\textwidth}
%        \centering
%
%\begin{tikzpicture}[ scale=1]
%
%\draw (0,0)--(10,0);
%\draw (0,1)--(10,1);
%\draw (0,2)--(10,2);
%\draw (0,4)--(10,4);
%\draw (0,5)--(10,5);
%
%\draw (0,0)--(0,5);
%\draw (1,0)--(1,5);
%\draw (2,0)--(2,5);
%\draw (3,0)--(3,5);
%
%\draw (5,0)--(5,5);
%\draw (6,0)--(6,5);
%\draw (7,0)--(7,5);
%
%\draw (10,0)--(10,5);
%\draw (11,0)--(11,5);
%
%
%%\node [blue] (a) at (-0.5,1.5) {$\bullet$};
%%\node [blue] (a) at (2.5,1.5) {$\bullet$};
%%\node[blue] at (2.8,1.5) {$u_2$};
%%
%%\draw[line width=0.35mm,red] (0,0)--(2,0);
%%\node [red] (a) at (0,0) {$\bullet$};
%%\node [red] (a) at (2,0) {$\bullet$};
%%\node [red] (a) at (1,0.8) {$\bullet$};
%%\node [red] (a) at (1,0.5) {$I_3$};
%%\node [red] (a) at (1,-0.3) {$I_1$};
%%\node [red] (a) at (1,0) {$\bullet$};
%%
%%\draw[thick, dotted] (-0.5,1.5) -- (0,0);
%%\draw[thick, dotted] (-0.5,1.5) -- (1,0);
%%\draw[thick, dotted] (2.5,1.5) -- (2,0);
%%\draw[thick, dotted] (2.5,1.5) -- (1,0);
%
%
%
%\end{tikzpicture}
%\caption{
%}
%\label{figure_baker}
%\end{figure}
%

Now we can assume that we are solving \probKexactCycle{} on $(H,f,k)$, where 
$(H,f,k)\in \mathscr F$ 
%is an instance obtained after applying Lemma~\ref{lem:exactcyclepw} 
(here we rename the function $f|_{V(H)}$ with $f$ for ease of presentation).  
Now we prove that if there is a cycle of length $k$ in $H$, then there is a cycle $C$ of length $k$ 
in $H$ such that for any two cells $(i,j)$ and $(i',j')$ of $f$, the number of edges of $E(C)$ with one end point in 
$(i,j)$ and other $(i',j')$ is at most $5$.   

\begin{lemma}
\label{lemm:path:shortinteraction}
Let $(H,f: V(H)\rightarrow [2k]\times [2k],k)$ be a \Yes{} instance of \probKexactCycle.  Then 
there is a cycle $C$ of length $k$ in $H$ such that for any two distinct cells $(i,j)$ and $(i',j')$ of $f$, 
the number of edges of $E(C)$ with one end point in $(i,j)$ and other $(i',j')$ is at most $5$.
\end{lemma}
\begin{proof}
Let $C$ be a $k$ length cycle such that the number edges of $E(C)$ whose end points are in different cells is minimized. 
We claim that   for any two disjoint cells $(i,j)$ and $(i',j')$, the number of edges of $E(C)$ with one end point in 
$(i,j)$ and other $(i',j')$ is at most $4$. Suppose not. Then there exist $(i,j)$ and $(i',j')$ such that 
the number of edges of $E(C)$ with one end point in $(i,j)$ and other in $(i',j')$ is at least $6$. 
Let $C=P_1 [u_1v_1] P_2  [u_2v_2] P_3 [u_3v_3] P_4 [u_4v_4] P_5 [u_5v_5] P_6 [u_6v_6]$  
where for each $\{u_r,v_r\},r\in [6]$, one end point is in the cell $(i,j)$ and other in the cell $(i'j')$, and  each subpath $P_\ell, \ell \in [6]$, can be empty too. Since  $C$ is a cycle, at least 3 edges from $\{\{u_r,v_r\}~:~i\in [6]\}$ form a matching.  Let 
$\{u_{r_1},v_{r_1}\},\{u_{r_2},v_{r_2}\}$ and $\{u_{r_3},v_{r_3}\}$ be a matching of size $3$, where $\{r_1,r_2,r_3\}\subseteq [6]$. 
Then, by pigeon hole principle there exist $r,r' \in \{r_1,r_2,r_3\}$ such that either $u_r,u_{r'} \in f^{-1}{(i,j)}$ or 
$u_r,u_{r'} \in f^{-1}{(i',j')}$. Without loss of generality assume that  $u_r,u_{r'} \in f^{-1}{(i,j)}$ (otherwise 
we rename cell $(i,j)$ with $(i',j')$ and vice versa).  That is, $C=[u_rv_r]Q_1[u_{r'}v_{r'}]Q_2$ 
such that $u_r,u_{r'}\in f^{-1}{(i,j)}$ and $v_r,v_{r'}\in f^{-1}{(i',j')}$.  
%of $P=Q_1[u_rv_r]Q_2[u_{r'}v_{r'}]Q_3$. 
Then, since $f^{-1}{(i,j)}$ and $f^{-1}{(i',j')}$ are cliques,   
$C'=[u_ru_{r'}] \overleftarrow{Q}_1 [v_rv_{r'}] Q_2$  is a $k$ length cycle in $G$,  such that 
the number edges of $E(C')$ whose end points are in different cells is less than that of $E(C)$, 
which is contradiction to our assumption.    
See Fig.~\ref{figure_cycle_replacement} for  an illustration of $C$ and $C'$. This completes 
the proof of the lemma. 
\end{proof}

\begin{figure}
%\begin{subfigure}[b]{0.5\textwidth}
        \centering

\begin{tikzpicture}[ scale=1]

\draw (0,0) ellipse (0.5 cm and 1.5cm); 
\draw (2,0) ellipse (0.5 cm and 1.5cm);
\node  (a) at (0,0.5) {$\bullet$};
\node  (a) at (0,-0.5) {$\bullet$};
\node  (a) at (0,0.75) {$u_r$};
\node  (a) at (0,-0.75) {$u_{r'}$};
\node  (a) at (2,0.5) {$\bullet$};
\node  (a) at (2,-0.5) {$\bullet$};
\node  (a) at (2,0.75) {$v_r$};
\node  (a) at (2,-0.75) {$v_{r'}$};

\draw[line width=0.35mm,red] (0,0.5)--(2,0.5);
\draw[line width=0.35mm,red] (0,-0.5)--(2,-0.5);
%\draw (2,0.5) .. controls (5,0) and (0,-3) .. (0,-0.5);
\draw[thick]  plot [smooth] coordinates {(2,0.5) (4,0) (4,-1) (2,-2) (-1,-2) (-1,-0.5)(0,-0.5)};
\draw[thick, dotted]  plot [smooth] coordinates {(2,-0.5) (3,0) (4,1) (2,2) (-1,2) (-1,0.5)(0,0.5)};

\draw (7,0) ellipse (0.5 cm and 1.5cm); 
\draw (9,0) ellipse (0.5 cm and 1.5cm);
\draw (0,0) ellipse (0.5 cm and 1.5cm); 
\draw (2,0) ellipse (0.5 cm and 1.5cm);

\node  (a) at (7,0.5) {$\bullet$};
\node  (a) at (7,-0.5) {$\bullet$};
\node  (a) at (9,0.5) {$\bullet$};
\node  (a) at (9,-0.5) {$\bullet$};
\node  (a) at (7,0.75) {$u_r$};
\node  (a) at (7,-0.75) {$u_{r'}$};
\node  (a) at (9,0.75) {$v_r$};
\node  (a) at (9,-0.75) {$v_{r'}$};

\draw[line width=0.35mm,green] (7,0.5)--(7,-0.5);
\draw[line width=0.35mm,green] (9,0.5)--(9,-0.5);
\draw[thick]  plot [smooth] coordinates {(9,0.5) (11,0) (11,-1) (9,-2) (6,-2) (6,-0.5)(7,-0.5)};
\draw[thick, dotted]  plot [smooth] coordinates {(9,-0.5) (10,0) (11,1) (9,2) (6,2) (6,0.5)(7,0.5)};

%\node [blue] (a) at (-0.5,1.5) {$\bullet$};
%\node [blue] (a) at (2.5,1.5) {$\bullet$};
%\node[blue] at (2.8,1.5) {$u_2$};
%
%\draw[line width=0.35mm,red] (0,0)--(2,0);
%\node [red] (a) at (0,0) {$\bullet$};
%\node [red] (a) at (2,0) {$\bullet$};
%\node [red] (a) at (1,0.8) {$\bullet$};
%\node [red] (a) at (1,0.5) {$I_3$};
%\node [red] (a) at (1,-0.3) {$I_1$};
%\node [red] (a) at (1,0) {$\bullet$};
%
%\draw[thick, dotted] (-0.5,1.5) -- (0,0);
%\draw[thick, dotted] (-0.5,1.5) -- (1,0);
%\draw[thick, dotted] (2.5,1.5) -- (2,0);
%\draw[thick, dotted] (2.5,1.5) -- (1,0);

\end{tikzpicture}
\caption{Illustration of Lemma~\ref{lemm:path:shortinteraction}.
Figure on the left is the cycle $C=[u_rv_r]Q_1[u_{r'}v_{r'}]Q_2$ 
and the  one on the right is the cycle 
$C'=[u_ru_{r'}] \protect\overleftarrow{Q}_1 [v_rv_{r'}] Q_2$}
\label{figure_cycle_replacement}
\end{figure}
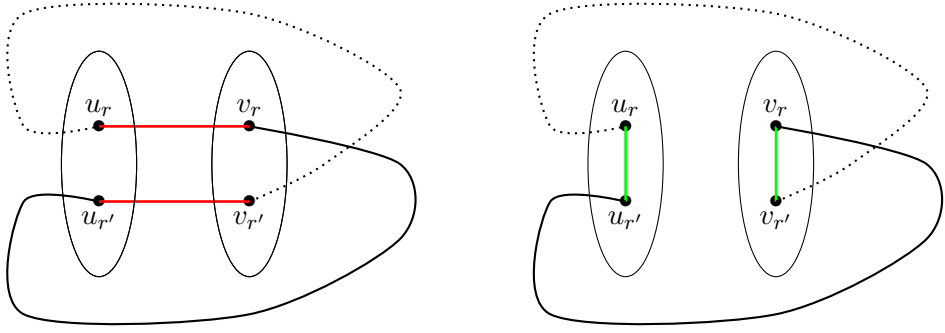

%Before explaining our dynamic programming algorithm over path decomposition, we explain {\em nice path decomposition}, 
%which eases the presentation of the algorithm.
%\begin{definition} 
%A path decomposition $\PP'=(X_1,\ldots,X_r)$ of a \gcgraph{} $H$ with a representation $f : V(H) \rightarrow [t]\times [t']$ 
%is called a nice path decomposition if 
%\begin{itemize}
%\item $X_1=X_r=\emptyset$, and 
%\item for every $\ell \in [r-1]$, there is a cell $(i,j)$ of $f$ such either $(i)$ $f^{-1}(i,j)\cap X_{\ell}=\emptyset$ and  
%$X_{\ell+1}=X_{\ell} \cup f^{-1}(i,j)$ or $(ii)$ $f^{-1}(i,j)\subseteq X_{\ell}$ and  
%$ X_{\ell+1}=X_{\ell}\setminus f^{-1}(i,j)$. 
%\end{itemize}
%The bags of the form $X_{\ell+1}=X_{\ell} \cup f^{-1}(i,j)$ is called {\em introduce cell bag} and 
%bags of the form $ X_{\ell+1}=X_{\ell}\setminus f^{-1}(i,j)$ is called {\em forget cell bag}. 
%\end{definition}

%\begin{lemma}
%\label{lem:nicepw}
%Given a \gcgraph{} $G$ along with a representation $f : V(G) \rightarrow [t]\times [t']$ and 
%a path decomposition $\PP$ of width $p$, one can compute, in polynomial time, a nice path 
%decomposition $\PP'$ of width $p$ such that each bag $X'$ in $\PP'$ is a subset of a bag $X$ in $\PP$. 
%\end{lemma}

Next we design a DP algorithm that  finds a cycle of length $k$, if it exists, satisfying properties of  Lemma~\ref{lemm:path:shortinteraction}.
% to look for cycle $C$ that satisfies 

\begin{lemma}
\label{lem:exact_k_cycle_restricted}
Let $(H,f: V(H)\rightarrow [2k]\times [2k], k)\in \mathscr F$ be an instance of \probKexactCycle\ 
and  $\PP$ be a \NCPD{$7\sqrt k$} of $H$. 
%such that $H$ has a 
%\NCPD{$7\sqrt k$} $\PP$ 
%nice path decomposition $\PP$ with each bag is a union of vertices from $7 \sqrt k$ many cells of $f$ 
%and $\vert V(H)\vert =k^{\OO(1)}$.  
 Then, given  $(H,f: V(H)\rightarrow [2k]\times [2k], k)$ and $\PP$, there is an algorithm ${\cal A}$ which 
runs in time  $2^{\OO(\sqrt k\log k)}$, and outputs \Yes, if there is 
a cycle $C$ in $H$ such that for any two distinct cells $(i,j)$ and $(i',j')$ of $f$, the number of edges with one end point in 
$(i,j)$ and other $(i',j')$ is at most $5$. Otherwise algorithm ${\cal A}$ will output \No. 
\end{lemma}
\begin{proof}
Algorithm ${\cal A}$ is a DP algorithm over the 
%nice path decomposition 
\NCPD{$7\sqrt k$}
$\PP=(X_1,\ldots X_{q})$ of $H$. For any $\ell\in [q]$, we define 
$H_{\ell}$ be the induced subgraph $H[\bigcup_{i\leq \ell} X_i]$ of $H$. 
%Algorithm ${\cal A}$ first guesses two vertices end vertices $s,t$ which are two end points of a path $P$, 
%satisfying the property of {lemm:path:shortinteraction}. Now add an edge $\{s,t\}$ and modify the path 
%decomposition $\PP$ by adding $\{s,t\}$ to each bag. 
%For the ease of presentation let $\PP=(X_1,\ldots X_{q})$ be the path decomposition of $H+\{\{s,t\}\}$.  
%Now we design a dynamic programming algorithm over the path decomposition $\PP$, which 
%will output a $k$ length cycle $C$ containing the edge $\{s,t\}$ where  for any two disjoint cells $(i,j)$ and $(i',j')$ of $g$, the number of edges with one end point in 
%$(i,j)$ and other $(i',j')$ is at most $4$, if such a cycle exists in the graph $H+\{\{s,t\}\}$.   
%
%%Now we define the states of the dynamic programming algorithm. 
%Notice 
%that $\PP=(X_1,\ldots X_{q})$ is a path decomposition such that 
%$(X_1\setminus \{s,t\},\ldots X_{q}\setminus \{s,t\})$ is a nice path 
%decomposition of $H-\{s,t\}$.  
%
%For each $\ell \in [q]$ and  a subset $Z=\{\{z_1,z_1'\},\ldots, \{z_s,z_s'\}\}$ of pairs of vertices of $X_{\ell}$ of cardinality 
%at most $350 \sqrt k +1$ and $F\subseteq X_{\ell}$, of cardinality at most $700\sqrt k+2$,  we define the following value. 
%For $ \emptyset \neq Z =\{\{z_1,z_1'\},\ldots, \{z_s,z_s'\}\}$.
% and $F \subseteq {X_{\ell}}$ of cardinality at most ${700\sqrt k+2}$,   
Define ${\cal C}$ to be the set of $k$ length cycles in $H$ such that for any 
$C\in{\cal C}$ and  two disjoint cells $(i,j)$ and $(i',j')$ of $f$, the number of edges of $E(C)$ with one end point in 
$(i,j)$ and other $(i',j')$ is at most $5$.
Let $C\in {\cal C}$. 
%Because of Claim~\ref{lemm:pwfewcells} 
%%and \ref{lem:nicepw} 
% and 
%the fact that 
%for any two disjoint cells $(i,j)$ and $(i',j')$ of $f$, the number of edges with one end point in 
%$(i,j)$ and other $(i',j')$ is at most $4$, we have that for any bag $X_{\ell}$ of $\PP$, the number of 
%vertices of $V(C)\cap X_{\ell}$ which has neighbours in $V(H)\setminus X_{\ell}$ is bounded by 
%$\OO(\sqrt k)$. 
Since $\PP$ is a \NCPD{$7\sqrt k$} and the fact 
that for any two distinct cells $(i,j)$ and $(i',j')$ of $f$, the number of edges of $C$ with one end point in 
$(i,j)$ and other $(i',j')$ is at most $5$, we have that for any bag $X_{\ell}$ of $\PP$, the number of 
vertices of $V(C)\cap X_{\ell}$ which has a neighbour in $V(H)\setminus X_{\ell}$ is bounded by 
$\OO(\sqrt k)$. 
This allows us to keep only $2^{\OO(\sqrt k\log k)}$ states in the DP 
algorithm.  Fix any $\ell\in [q]$ and define 
${C}_L$ the set of paths  of $C$ (or the cycle $C$ itself) when we restrict 
$C$ to $H_{\ell}$. That is ${C}_L=H_{\ell}[E({C})]$.  
Let $\widehat{C}_L=\{\{u,v\}~|~ \mbox{there is a $u$-$v$ path in } {C}_L\}$. 
Notice that $\bigcup_{P\in \widehat{C}_L} P$  is the set of vertices of degree $0$ or $1$  
in  ${C}_L$ and  $\bigcup_{P\in \widehat{C}_L} P \subseteq X_{\ell}$. 
Since $X_{\ell}$ is  a union of vertices from at most $7 \sqrt k$ many cells of $f$ and 
for any two distinct cells $(i,j)$ and $(i',j')$ of $f$, the number of edges of $E(C)$ with one end point in 
$(i,j)$ and other $(i',j')$ is at most $5$, 
and by property $2$ of the \gcgraph, 
we have that the cardinality of $\bigcup_{P\in \widehat{C}_L} P$  
is at most $5\cdot 24 \cdot 7 \sqrt k= 840\sqrt k$. In our DP algorithm 
we will have state indexed by $(\ell,\widehat{C}_L,\vert E(C_L) \vert)$  which 
will be set to $1$.   
%where the corresponding entry will be 
%$1$ if there is a set {\cal Q} of vertex disjoint paths where the end points of the paths are specified by 
%$\widehat{C}_L$ and $\vert E({\cal Q})\vert = k'$. 
%
Formally, for any $\ell\in [q]$, $k'\in [k]$ and a family $\ZZ$ of vertex disjoint sets of size at most $2$ of $X_{\ell}$ with 
the property that the cardinality of $\bigcup_{Z\in \ZZ} Z$ is at most  $840\sqrt k$, we will have a DP table 
entry ${\cal A}[\ell, \ZZ, k']$. For each $\ell \in [q]$, we maintain the following correctness invariant. 

\begin{quote}
%\item 
{\bf Correctness Invariants:}
$(i)$ For every $C\in {\cal C}$, let ${C}_L =H_{\ell}[E(C)]$ and $\widehat{C}_L =$\\
$ \{\{u,v\} ~|~ \mbox{there is a connected component $P$  in ${C}_L$  and $P$ is a $u$-$v$ path as well} \}$. 
Then ${\cal A}[\ell,\widehat{C}_L, \vert E(C_L) \vert]=1$, 
$(ii)$ for any 
family $\ZZ$ of vertex disjoint sets of size at most $2$ of $X_{\ell}$ with  $0<\vert \bigcup_{Z\in \ZZ} Z \vert \leq 840\sqrt k$, 
$k'\in [k]$, and ${\cal A}[\ell,\ZZ,k']=1$, there is a set ${\cal Q}$ of $\vert \ZZ \vert$ vertex disjoint paths in $H_{\ell}$
where the end points of each path are specified by a set in $\ZZ$ and $\vert E(\QQ)\vert=k'$, 
and $(iii)$  If ${\cal A}[\ell,\emptyset ,k]=1$, then there is a cycle of length $k$ 
%if and only if there is a cycle $C\in \CC$ which is fully contained 
in $H_{\ell}$. 
\end{quote}

The correctness of the our algorithm will follow from the correctness invariant. 
Before explaining the computation of DP table entries, we first 
define some notations. Fix $\ell\in [q]$.  
For any $C\in {\cal C}$,  
define ${C}_L=H_{\ell}[E(C)]$ and $C_R=H[E(C)\setminus E(C_{L})]$.
For any family $\QQ$ of vertex disjoint paths with end points in $X_{\ell}$, 
define $\widehat{\QQ}=\{\{u,v\} ~|~ \mbox{there is a $u$-$v$ path in } \QQ\}$. 
That is, for any $C\in \CC$, $\widehat{C}_L =\{\{u,v\} ~|~ \mbox{there is a $u$-$v$ path in } {C}_L\}$. 
% For any family $\ZZ$ of vertex disjoint sets of size at most $2$, we define 
%$\ZZ'$
%The following lemma is useful to while proving the correctness invariants in each step. 
%\begin{lemma}
%\label{lem:cycle_eqclass}
%Let $C\in \CC$ and  $\QQ$ of vertex disjoint paths with end points in $X_{\ell}$ such that 
%$\widehat{C}_L=\widehat{\QQ}$, then $E(C_R)\cup E(\QQ)$ for a cycle of length $k$. 
%\end{lemma}
%\begin{proof}
%\todo[inline]{add a proof}
%\end{proof}
Now we explain how to compute the values ${\cal A}[\cdot,\cdot,\cdot]$. 
In what follows we explain how to compute   ${\cal A}[\ell,\ZZ,k']$
for every $\ell \in [q]$, $k'\in [k]$ and 
family $\ZZ$ of vertex disjoint sets of size at most $2$ of $X_{\ell}$ with $\vert \bigcup_{Z\in\ZZ} Z\vert \leq 840\sqrt k$, 
the running to compute it from the previously computed DP table entries, and  prove the correctness 
invariants. While proving the correctness invariants for $\ell$, we assume that the correctness 
invariant holds for $\ell-1$.  
When $\ell=1$, $X_1=\emptyset$ and the DP table entries are defined as follows.

\begin{eqnarray*}
{\cal A}[1,\emptyset,k'] =
\left\{ \begin{array}{rl}
1 & \mbox{if }  k'=0\\ 
0 & \mbox{otherwise} 
\end{array}\right.
\end{eqnarray*}
Since $H_1=(\emptyset,\emptyset)$, the correctness invarient follows. The values ${\cal A}[1,*,*]$ can be computed 
in $\OO(1)$ time.  Now we move to the case where $\ell>1$.

\subparagraph*{Case 1:  $X_{\ell}$ is an Introduce  bag.}  That is $X_{\ell}= X_{\ell-1} \cup  f^{-1}(i,j)$ 
for some cell $(i,j)$. 
Fix a family $\ZZ$ of vertex disjoint sets of size at most $2$ of $X_{\ell}$ such that $\vert \bigcup_{Z\in \ZZ} Z\vert$ is at most  $840\sqrt k$ and $k'\in [k]$.  Define $\QQ_{\ell}(r)$ to be the set $Q$ of vertex disjoint paths in $H[X_{\ell}\setminus X_{\ell-1}]=H[f^{-1}(i,j)]$ with at most $120$ end points 
and $\vert E(\QQ)\vert=r$.  
Recall that that   $\widehat{Q}=\{\{u,v\} ~|~ \mbox{there is a $u$-$v$ path in } Q\}$. 
Define $\widehat{\QQ}_{\ell}(r)=\{\widehat{Q}~:~ Q \in {\QQ}_{\ell}(r)\}$. 

\begin{claim}
\label{claim:cycle_clique_enu}
$\vert \widehat{\QQ}_{\ell}(r)\vert = k^{\OO(1)}$ and 
$\widehat{\QQ}_{\ell}(r)$ can be enumerated in time  $k^{\OO(1)}$.  
\end{claim}
\begin{proof}
We know that $H[X_{\ell}\setminus X_{\ell-1}]$ is a clique on $f^{-1}(i,j)$. 
%Hence, any set $Q$ of vertex disjoint paths in $H[X_{\ell}\setminus X_{\ell-1}]$ can be specified by 
%$\widehat {Q}$ of cardinality at most $120$. In other words, 
For any family $\WW$ of 
vertex disjoint sets of size at most $2$ of $f^{-1}(i,j)$, one can get  $\vert \WW\vert$ vertex disjoint paths 
with end points being the one specified by the sets in $\WW$ with total number of edges $r$, only if 
either $\vert \{W\in \WW~:~\vert W \vert=2\}\vert = r$ or $1\leq \vert \{W\in \WW~:~\vert W \vert=2\}\vert < r$    
and $\vert \{W\in \WW~:~\vert W \vert=2\}\vert + \vert f^{-1}(i,j) \setminus  \bigcup_{W\in \WW} W\vert \geq r$. 
Hence we can enumerate   $\widehat{\QQ}_{\ell}(r)$ in time  $k^{\OO(1)}$
\end{proof}

For any family ${\cal Y}$ of sets of size at most $2$, define ${\cal Y}'=\{A\in {\cal Y}~:~ \vert A \vert =2\}$. 
For any three families ${\cal Y}_1, {\cal Y}_2$ and ${\cal Y}_3$ of sets of size at most $2$ in $X_{\ell}$, 
we say that ${\cal Y}_1\cup {\cal Y}_2 \cup {\cal Y}_3$ for a family of paths (respectively, a cycle) in $G'=K[X_{\ell}]$, 
if the graph $(\bigcup_{\substack{i\in[3]\\ Y\in {\cal Y}_i}} Y, \bigcup_{i\in [3]} {\cal Y}_i')$ form a family of paths (respectively, a cycle). 
Consider the case when $\ZZ\neq\emptyset$.
\begin{equation}
\mbox{If 
$ {\ZZ}\in  \widehat{\QQ}_{\ell}(k')$, then we set } {\cal A}[\ell,\ZZ,k']=1  \label{eng:cycle:intro1a}
\end{equation}
Otherwise, 
\begin{eqnarray}
{\cal A}[\ell,\ZZ,k'] =\max & \Big\{& {\cal A}[\ell-1,\ZZ',k'']~:~ \ZZ'\neq \emptyset \mbox{ and there exist } r\in {\mathbb N}, \widehat{\QQ}\in  \widehat{\QQ}_{\ell}(r), \nonumber\\
&&  E'\subseteq E\left (\bigcup_{Q\in \widehat{\QQ}} Q, \bigcup_{Z\in \widehat{\ZZ'} }Z\right) 
\mbox{ such that },     
\vert E'\vert \leq 120, k'=k'' +r+\vert E'\vert,  \nonumber \\
&&  \ZZ'\cup  \widehat{\QQ} \cup E'  \mbox{ forms a set ${\RR}$ of paths in } K[X_{\ell}],  
\mbox{ and }  \widehat{\cal R}=\ZZ\Big\}  \label{eng:cycle:intro1}
\end{eqnarray}

Now consider that case when $\ZZ=\emptyset$. 

\begin{equation}
\mbox{If 
${\cal A}[\ell-1,\emptyset,k']=1$, then we set } {\cal A}[\ell,\emptyset,k']=1  \label{eng:cycle:intro2a}
\end{equation}
%If $ {\ZZ}\in  \widehat{\QQ}_{\ell}(k')$, then we set 
%${\cal A}[\ell,\ZZ,k']=1$. 

Otherwise 

\begin{eqnarray}
{\cal A}[\ell,\emptyset,k'] =\max & \Big\{& {\cal A}[\ell-1,\ZZ',k'']~:~\mbox{there exist } r\in {\mathbb N}, \widehat{\QQ}\in  \widehat{\QQ}_{\ell}(r), \nonumber\\
&& E'\subseteq E\left (\bigcup_{Q\in \widehat{\QQ}} Q, \bigcup_{Z\in \widehat{\ZZ'} }Z\right) 
\mbox{ such that }, \vert E'\vert \leq 120,  \nonumber \\
&&   k'=k'' +r+\vert E'\vert, \mbox{ and } \ZZ'\cup  \widehat{\QQ} \cup E'  \mbox{ forms a cycle in } K[X_{\ell}] \Big\}
\label{eng:cycle:intro2}
\end{eqnarray}

Notice that in the above computation (Equations~\ref{eng:cycle:intro1} and \ref{eng:cycle:intro2}) the number of potential choices for $\ZZ'$ is bounded by $2^{\OO(\sqrt k\log k)}$. 
By Claim~\ref{claim:cycle_clique_enu} we know that the cardinality of $ \widehat{\QQ}_{\ell}(r)$  is at most $k^{\OO(1)}$ and 
it can be enumerated in time  $k^{\OO(1)}$. Since $\vert X_{\ell} \vert =k^{\OO(1)}$, the number of choices 
for $E'$ is the above computation is bounded by $k^{\OO(1)}$. This implies that we can compute 
${\cal A}[\ell,\ZZ,k']$ using previously computed DP entries in time $2^{\OO(\sqrt k\log k)}$. 

Before proving the correctness invariant, we state the following simple claim, 
which can be proved using induction of $\ell$. 

\begin{claim}
\label{claim:claim:cycle:empty}
For any $\ell\in [q]$, ${\cal A}[\ell,\emptyset, 0]=1$. 
\end{claim}

Now we prove the correctness invariants. 
Let $C\in {\cal C}$. Recall that  ${C}_L=H_{\ell}[E(C)]$ and 
$\widehat{C}_L =
%\{\{u,v\} ~|~ \mbox{there is a $u$-$v$ path in } {C}_L\}$. 
\{\{u,v\} ~|~ \mbox{there is a connected component $P$  in ${C}_L$  and $P$ is a $u$-$v$ path as well} \}$. 
%Then ${\cal A}[\ell,\widehat{C}_L, \vert E(C_L) \vert]=1$. 
Partition the edges ${C}_L$ into $E_1\uplus E_2 \uplus E'$ where 
$E_1=E(C_L)\cap E(H_{\ell-1})$, $E_2= E(C_L)\cap E(X_{\ell}\setminus X_{\ell-1})$ 
and $E'=E(C_L)\cap E(X_{\ell-1}, X_{\ell}\setminus X_{\ell-1})$. Let $C_L'= H[E_1]$. That is, 
$C_L'=H_{\ell-1}[E(C)]$. 
Let $\widehat{C}_L' =\{\{u,v\} ~|~ \mbox{there is a $u$-$v$ path in } {C}_L'\}$.
%Let $\QQ=H[E_2]$.  
We have two cases based whether $\widehat{C}_L=\emptyset$ or not.  

Suppose $\widehat {C}_L\neq \emptyset$. If $C_L$ is a subgraph of $H[X_{\ell}\setminus X_{\ell-1}]$, 
then $\widehat{C}_L \in \widehat{Q}(\vert E(C_L)\vert)$.  
Then by Equation~\ref{eng:cycle:intro1a}, we have that ${\cal A}[\ell,\widehat{C}_L, \vert E(C_L) \vert]=1$. 
So now we have that $C_L$ is not a subgraph of $H[X_{\ell}\setminus X_{\ell-1}]$. 
This implies that either $E_1\neq \emptyset$ or $E_2\neq \emptyset$. In either case, 
we have that $\widehat{C}_L'\neq \emptyset$.
Since $X_{\ell-1}$ is a union of vertices from $7\sqrt k$ cells, the property $2$ of \gcgraph{} and $C\in \CC$, 
we have that the number of edges with one endpoint $X_{\ell-1}$ and other in $H\setminus X_{\ell-1}$ is 
at most $5\cdot 24 \cdot 7 \sqrt k = 840 \sqrt k$. This implies that 
$\vert \bigcup_{D\in \widehat{C}'_L} D \vert \leq 840 \sqrt k$. 
By the correctness invariant of statement $(i)$ for $\ell-1$, we have that 
$(a)$ ${\cal A}[\ell-1,\widehat{C}'_L, \vert E(C_L') \vert]=1$. 
Consider the graph $H[E_2]$. The graph 
$H[E_2]$ is a collection ${\cal Q}$ of paths in $H[X_{\ell}\setminus X_{\ell-1}]= K[f^{-1}(i,j)]$. 
Since $C\in \CC$, the number of edges of $E(C)$ with one end point in $f^{-1}(i,j)$ and other 
in $X_{\ell-1}$ is at most $4\cdot 25=120$.
This implies that $(b)$ $\vert E'\vert \leq 120$  
and $(c)$ $\widehat{\QQ}\in \widehat{Q}_{\ell} (r)$, where $r=\vert E(\QQ)\vert$. 
By facts $(a)$, $(b)$ and $(c)$, using Equation~\ref{eng:cycle:intro1}, 
we get ${\cal A}[\ell,\widehat{C}_L, \vert E(C_L) \vert]=1$. 

Now consider the case when $\widehat {C}_L= \emptyset$. In this case either 
 $E(C)\cap H_{\ell}=\emptyset$ or $E(C)\subseteq E(H_{\ell})$. If  $E(C)\cap H_{\ell}=\emptyset$, 
then $\vert E(C_L)\vert =0$ and hence ${\cal A}[\ell,\widehat{C}_L, \vert E(C_L) \vert]= {\cal A}[\ell,\emptyset, 0]=1$, 
by Claim~\ref{claim:claim:cycle:empty}. 
Now we have $\widehat {C}_L= \emptyset$ and $E(C)\subseteq E(H_{\ell})$. 
This implies that $\vert E(C_L)\vert =k$. 
If $E(C_L)=E(C_L')$, then ${\cal A}[\ell-1, \widehat{C}'_{L},k] ={\cal A}[\ell-1, \emptyset,k]=1$, 
by the statement $(i)$ of the correctness invariant for $\ell-1$. 
So, now we have $E(C_L)\neq E(C'_L)$. This implies that either $E_2$ or $E'$ is not empty. 
Since $\vert f^{-1}(i,j)\vert <k$ and $C_L=C$ form a cycle and  $E(C_L)\neq E(C'_L)$, 
we have that $\widehat{C}'_L\neq \emptyset$. 
The graph 
$H[E_2]$ is a collection ${\cal Q}$ of paths in $H[X_{\ell}\setminus X_{\ell-1}]= K[f^{-1}(i,j)]$. 
%Since $C\in \CC$, the number of edges of $E(C)$ with one end point in $f^{-1}(i,j)$ and other 
%in $X_{\ell-1}$ is at most $4\cdot 25=120$. 
As like before, we can bound 
$(d)$ $\vert \bigcup_{D\in \widehat{C}'_L} D \vert \leq 840 \sqrt k$, 
$(e)$ $\vert E'\vert \leq 120$  
and $(g)$ $\widehat{\QQ}\in \widehat{Q}_{\ell} (r)$, where $r=\vert E(\QQ)\vert$. 
By $(d)$ and the statement $(i)$ of the correctness invariant for $\ell-1$, 
we get $(h)$ ${\cal A}[\ell-1,\widehat{C}'_L, \vert E(C_L') \vert]=1$. 
By facts $(h)$, $(e)$ and $(d)$, using Equation~\ref{eng:cycle:intro2}, 
we get ${\cal A}[\ell,\widehat{C}_L, \vert E(C_L) \vert]={\cal A}[\ell,\emptyset, k]=1$. 
This completes the proof of statement $(i)$.

% 
%------------------------------------------
%
%
%
%
%
%
%
%Suppose $\widehat{C}_L'=\emptyset$. 
%%Then $E'=\emptyset$ and $E(C)=E_2$. 
%Then ${\cal A}[\ell,\widehat{C}_L, \vert E(C_L) \vert]= {\cal A}[\ell-1,\emptyset, 0]=1$,
%
%
% 
%by Claim~\ref{claim:claim:cycle:empty}. So now on,  
% $C_L$ is not empty.  
%If $C_L'=(\emptyset,\emptyset)$, then we have that $E'=\emptyset$, because $E'\subseteq E(X_{\ell-1}, X_{\ell}\setminus X_{\ell-1})$ 
%This implies that when $C_L'=(\emptyset,\emptyset)$, then $E(C_L) \subseteq E(X_{\ell}\setminus X_{\ell-1})$ and 
%hence the graph $H[E_1\cup E_2\cup E']=H[E_2]$ is a collection ${\cal Q}$ of paths in $H[X_{\ell}\setminus X_{\ell-1}]= K[f^{-1}(i,j)]$
%such that $\widehat{Q}=\widehat{C}_L$, This implies that ${\cal A}[\ell,\widehat{C}_L, \vert E(C_L) \vert]=1$, by Equation~\ref{eng:cycle:intro1a}.
%
% 
%So now we assume that $C_L'\neq (\emptyset,\emptyset)$. 
%Since $C_L'\neq (\emptyset,\emptyset)$ 
%
%
%
%
%
%
%
%
%This implies that   $\vert \bigcup_{D\in  \widehat{C}_L} D\vert \leq 840 \sqrt k$. 
%We have two cases based on whether $\vert E(C_L) \vert=k$ or not. When 
%$\vert E(C_L) \vert=k$, by Equation \ref{eng:cycle:intro2},  we get ${\cal A}[\ell,\widehat{C}_L, \vert E(C_L) \vert]=1$. 
%Now consider the case when $\vert E(C_L) \vert\neq k$. 
%In this case, since $C$ is a cycle, $\widehat{C}_L\neq \emptyset$.   
%Then, by  Equation~\ref{eng:cycle:intro1} ,  we get ${\cal A}[\ell,\widehat{C}_L, \vert E(C_L) \vert]=1$. 
%This completes the proof of statement $(i)$. 
% 

Now we need to prove statement $(ii)$ in the correctness invariants. 
Let $\ZZ$ be a family of vertex disjoint sets of size at most $2$ of $X_{\ell}$ with  $0<\vert \bigcup_{Z\in \ZZ} Z \vert \leq 840\sqrt k$
and $k'\in [k]$. 
Notice that $\ZZ\neq \emptyset$. 
Suppose in the above computation we set ${\cal A}[\ell,\ZZ,k']=1$.  
Either ${\cal A}[\ell,\ZZ,k']$ is set to $1$ because of Equation~\ref{eng:cycle:intro1a} or 
because of Equation~\ref{eng:cycle:intro1}. 
If  ${\cal A}[\ell,\ZZ,k']$ is set to $1$ because of Equation~\ref{eng:cycle:intro1a}, then 
we know that $ {\ZZ}\in  \widehat{\QQ}_{\ell}(k')$. By the definition of 
$\widehat{\QQ}_{\ell}(k')$, we get that there is a set ${\cal R}$ of vertex disjoint paths in $H[X_{\ell}\setminus X_{\ell-1}]$, 
hence in $H[X_{\ell}]$ and $\widehat{ \cal R}=\ZZ$.  
So, now assume that  ${\cal A}[\ell,\ZZ,k']$ is set to $1$ because of Equation~\ref{eng:cycle:intro1}. 
This implies that there exist $k'',r\in {\mathbb N}$,  a family $\ZZ'$ of vertex disjoint sets of size at most $2$ of $X_{\ell-1}$, 
$\widehat{\QQ}\in  \widehat{\QQ}_{\ell}(r)$,
and $E'\subseteq E\left (\bigcup_{Q\in \widehat{\QQ}} Q, \bigcup_{Z\in \widehat{\ZZ'} }Z\right)$ 
such that  ${\cal A}[\ell-1,\ZZ',k'']=1$, $\vert E'\vert \leq 120, k'=k'' +r+\vert E'\vert$,   $\ZZ'\cup  \widehat{\QQ} \cup E'$ 
forms a set of paths ${\cal R}$  in  $K[X_{\ell}]$ with  $ \widehat{\cal R}=\ZZ$.  
Since ${\cal A}[\ell-1,\ZZ',k'']=1$, by the statement $(ii)$ of the correctness invariant for  $\ell-1$, we have that there is 
 a set ${\cal Y}$ of $\vert \ZZ' \vert$ vertex disjoint paths  in $H_{\ell-1}$ 
where the end points of each path are specified by a set in $\ZZ'$ and $\vert E(\YY)\vert=k''$.
Let $\QQ$ be the set of paths in $\QQ_{\ell}(r)$ corresponding to the set $\widehat{\QQ}$.  
Thus by replacing each edge of $\ZZ'$ in  $\ZZ'\cup  \widehat{\QQ} \cup E'$ by the corresponding path in 
$\YY$ and each edge of $ \widehat{\QQ}$ by a corresponding path in $\QQ$, 
 we can get a set $\WW$ of vertex disjoint paths in $H_{\ell}$, because the internal vertices of paths in $\YY$ are disjoint from 
$(X_{\ell} \setminus X_{\ell-1})\cup \bigcup_{Z\in \ZZ} Z $ and the interval vertices of paths in $\QQ$ are disjoint from 
$X_{\ell-1}\cup \bigcup_{Z\in \ZZ} Z$. 
More over,  $ \widehat{\WW}=\ZZ$ 
and $\vert E(\WW)\vert=\vert E(\YY)\vert +\vert E(\QQ)\vert+\vert E'\vert = k''+r+\vert E'\vert=k'$.  This 
completes the proof of statement $(ii)$ in the correctness invariant. 

Now we prove statement $(iii)$ in the correctness invariants. Suppose 
we set $A[\ell,\emptyset,k]=1$. 
Then either ${\cal A}[\ell,\emptyset,k]$ is set to $1$ because of Equation~\ref{eng:cycle:intro2a} or 
because of Equation~\ref{eng:cycle:intro2}.  
If  ${\cal A}[\ell,\emptyset,k]$ is set to $1$ because of Equation~\ref{eng:cycle:intro1a}, then 
we know that ${\cal A}[\ell-1,\emptyset,k]=1$, then by the statement $(iii)$ of the correctness invariant for $\ell-1$, 
we have that there is a cycle of length $k$ in $H_{\ell-1}$ and hence in $H_{\ell}$. 
Suppose ${\cal A}[\ell,\emptyset,k]$ is set to $1$ because of Equation~\ref{eng:cycle:intro2}. 
Then, there exist $k'',r\in {\mathbb N}$,  a family $\ZZ'$ of vertex disjoint sets of size at most $2$ of $X_{\ell-1}$, 
$\widehat{\QQ}\in  \widehat{\QQ}_{\ell}(r)$,
and $E'\subseteq E\left (\bigcup_{Q\in \widehat{\QQ}} Q, \bigcup_{Z\in \widehat{\ZZ'} }Z\right)$ 
such that  ${\cal A}[\ell-1,\ZZ',k'']=1$, $\vert E'\vert \leq 120, k=k'' +r+\vert E'\vert$,   $\ZZ'\cup  \widehat{\QQ} \cup E'$ 
forms a cycle  in  $K[X_{\ell}]$.  
Since ${\cal A}[\ell-1,\ZZ',k'']=1$, by the statement $(ii)$ of the correctness invariant for $\ell-1$, we have that there is 
 a set ${\cal Y}$ of $\vert \ZZ' \vert$ vertex disjoint paths  in $H_{\ell-1}$
 where the end points of each path are specified by a set in $\ZZ'$ and $\vert E(\YY)\vert=k''$.
Let $\QQ$ be the set of paths in $\QQ_{\ell}(r)$ corresponding to the set $\widehat{\QQ}$.  
Thus by replacing each edge of $\ZZ'$ in  $\ZZ'\cup  \widehat{\QQ} \cup E'$ by the corresponding path in 
$\YY$ and each edge of $ \widehat{\QQ}$ by a corresponding path in $\QQ$, 
 we can get a cycle $C$ in $H_{\ell}$, because the internal vertices of paths in $\YY$ are disjoint from $(X_{\ell} \setminus X_{\ell-1})\cup \bigcup_{Z\in \ZZ} Z $ and the interval vertices of paths in $\QQ$ are disjoint from 
$X_{\ell-1}\cup \bigcup_{Z\in \ZZ} Z$. 
More over,  
%$ \widehat{\WW}=\ZZ$ 
$\vert E(C)\vert=\vert E(\YY)\vert +\vert E(\QQ)\vert+\vert E'\vert = k''+r+\vert E'\vert=k$.  This 
completes the proof of statement $(iii)$ in the correctness invariant.

\subparagraph*{Case 2: $X_{\ell}$ is a forget bag.} 
%Suppose $X_{\ell}$ is a forget cell bag. That is $X_{\ell}= (X_{\ell-1} \setminus  f^{-1}(i,j))\cup \{s,t\}$ 
%for some cell $(i,j)$. 
Fix a family $\ZZ$ of vertex disjoint sets of size at most $2$ of $X_{\ell}$ such that $\vert \bigcup_{Z\in \ZZ} Z\vert$ is at most  $840\sqrt k$ and $k'\in [k]$.
\begin{equation}
{\cal A}[\ell,\ZZ,k']={\cal A}[\ell-1,\ZZ,k'] \label{eqn:cycle:forget}
\end{equation}
Clearly ${\cal A}[\ell,\ZZ,k']$ can be computed in $\OO(1)$ time using the previously computed DP table entries.

Now we prove the correctness invariant. 
Let $C\in {\cal C}$. 
Recall that  ${C}_L=H_{\ell}[E(C)]$ and $\widehat{C}_L =
%\{\{u,v\} ~|~ \mbox{there is a $u$-$v$ path in } {C}_L\}$. 
\{\{u,v\} ~|~ \mbox{there is a connected component $P$  in ${C}_L$  and $P$ is a $u$-$v$ path as well} \}$. 
By arguments similar to those in Case 1, we get $\vert \bigcup_{D\in  \widehat{C}_L} D\vert \leq 840 \sqrt k$.
Since $H_{\ell}=H_{\ell-1}$, we have that $C_L=H_{\ell-1}[E(C)]$. 
Hence by the correctness invariant for $\ell-1$, we have that 
${\cal A}[\ell-1,\widehat{C}_L, \vert E(C_L) \vert]=1$. 
Hence, by Equation~\ref{eqn:cycle:forget}, ${\cal A}[\ell-1,\widehat{C}_L, \vert E(C_L) \vert]=1$.    

Now we need to prove statement $(ii)$ of the correctness invariants.
Let $\ZZ$ be a family of vertex disjoint sets of size at most $2$ of $X_{\ell}$ with  $0<\vert \bigcup_{Z\in \ZZ} Z \vert \leq 840\sqrt k$
and $k'\in [k]$. 
Suppose in the above computation (Equation~\ref{eqn:cycle:forget}) we set ${\cal A}[\ell,\ZZ,k']=1$.  
This implies that ${\cal A}[\ell-1,\ZZ,k']=1$. 
Since ${\cal A}[\ell-1,\ZZ,k']=1$, by the correctness invariant for $\ell-1$, we have that there is 
 a set ${\cal Y}$ of $\vert \ZZ \vert$ vertex disjoint paths  in $H_{\ell-1}$ 
where the end points of each path are specified by a set in $\ZZ$ and $\vert E(\YY)\vert=k'$.
Since $H_{\ell}=H_{\ell-1}$, $\YY$ is the required set of vertex disjoint paths 
and this  completes the proof of statement $(ii)$ in the correctness invariant. 

Now we prove the statement $(iii)$ in the correctness invariants.
Suppose in the above computation (Equation~\ref{eqn:cycle:forget}) we set ${\cal A}[\ell,\emptyset,k]=1$.  
This implies that ${\cal A}[\ell-1,\emptyset,k]=1$
Since ${\cal A}[\ell-1,\emptyset,k]=1$, by the correctness invariant for $\ell-1$, we have that there is 
 $k$ length cycle in $H_{\ell-1}$, and hence in $H_{\ell}$. 
This completes the proof of correctness invariants.

Algorithm ${\cal A}$ output  \Yes{} if ${\cal A}[q,\emptyset,k]=\emptyset$ and a output \No{} otherwise.  
The correctness of the algorithm follows from the correctness invariants. Now we analyse the total 
running time. Notice that $\vert V(H)\vert = k^{\OO(1)}$ and number of DP table entries is bounded by    
$2^{\OO(\sqrt k\log k)}$. Each DP table entry can be computed in time  $2^{\OO(\sqrt k\log k)}$ using 
the previously stored values in the DP table. Hence the total running time of the 
algorithm is $2^{\OO(\sqrt k\log k)}$.  
\end{proof}

%Theorem~\ref{thm:exactcycle} follows from 
Lemmata~\ref{lem:turingKernISGrid}, \ref{lem:exactcyclepw}, \ref{lemm:path:shortinteraction} and \ref{lem:exact_k_cycle_restricted} implies the following Lemma.   

\begin{lemma}
\label{lem:exactKCycle}
\probKexactCycle{}  on \gcgraph{s}  
can be solved  in time  $2^{\OO(\sqrt k\log k)} n^{\OO(1)}$. 
%when $n$ is the number of vertices in the input graph. 
\end{lemma}

Theorem~\ref{thm:exactcycle} follows from Lemma~\ref{lem:exactKCycle} and Corollary~\ref{cor:firstPhaseToGrid}. 
We can design a similar algorithm to  solve {\probKPath} in time $2^{\OO(\sqrt k\log k)} n^{\OO(1)}$.

%!TEX root = main_unit_disk.tex

\section{Longest Cycle}\label{sec:longCyc}

In this section, we show that \probKCycle admits a subexponential-time parameterized algorithm. More precisely, we prove the following.

\begin{theorem}\label{thm:longestCyc}
\probKCycle on unit disk/square graphs can be solved in time $2^{\OO(\sqrt{k}\log k)}\cdot n^{\OO(1)}$.
\end{theorem}

We start by stating a direct implication of Lemma~\ref{lem:exactKCycle}.

\begin{corollary}\label{cor:nearKCycle}
Given a graph $G$, representation $f$ and $k\in\mathbb{N}$, it can be determined in time $2^{\OO(\sqrt{k}\log k)}\cdot n^{\OO(1)}$ whether $G$ contain a cycle whose number of vertices is between $k$ and $2k$.
\end{corollary}

\begin{proof}
Run the algorithm given by Lemma~\ref{lem:exactKCycle} with $k=\ell,\ell+1,\ldots,2\ell$, and return \Yes{} if and only if at least one of the executions returns \Yes. Correctness and running time follow directly from Lemma~\ref{lem:exactKCycle}.
\end{proof}

Next, we examine the operation that contracts an edge. To this end, we need the following. 

\begin{definition}\label{def:contractiblePair}
A pair $(u,v)$ of distinct vertices $u,v\in V(G)$ is {\em contractible} if $f(u)=f(v)$.
\end{definition}

Note that if $(u,v)$ is a contractible pair, then by Condition \ref{condition:GridClique1} in Definition \ref{def:GridClique}, it holds that $\{u,v\}\in E(G)$. Now, given a contractible pair $(u,v)$, denote $e=\{u,v\}$, and define the function $f_{/e}: V(G/e)\rightarrow[t]\times[t']$ as follows. For all $w\in V(G)\setminus\{v,u\}$, define $f_{/e}(w)=f(w)$. Moreover, define $f_{/e}(x_{\{u,v\}})=f(u)$. By Definitions \ref{def:GridClique} and \ref{def:cellGraph}, we immediately have the following.

\begin{observation}\label{obs:funcFe}
The function $f_{/e}$ is a representation of $G/e$. Furthermore, $G$ and $G/e$ have the same cell graph.
\end{observation}

In particular, we deduce that $(G/e,f_{/e},k)$ is an instance of \probKCycle on \gcgraphs. Next, we note that the operation that contracts an edge preserves the answer \No---the correctness of this claim follows from the fact that $G/e$ is a minor of $G$.

\begin{observation}\label{obs:longCycPreserveNo}
Let $(G,f,k)$ be an instance of \probKCycle on \gcgraphs. Then, $(G/e,f_{/e},k)$ is an instance of \probKCycle on \gcgraphs such that if $(G,f,k)$ is a \No-instance, then $(G/e,f_{/e},k)$ is a \No-instance.
\end{observation}

Now, we also verify that in case there exists a cycle on at least $2k$ vertices, the operation that contracts an edge also preserves the answer \Yes.

\begin{lemma}\label{lem:longCycPreserveYes}
Let $(G,f,k)$ be an instance of \probKCycle on \gcgraphs{} such that $G$ contains a cycle $C$ on at least $2k$ vertices. Then, $(G/e,f_{/e},k)$ is a \Yes-instance.
\end{lemma}

\begin{proof}
Denote $e=\{u,v\}$. In case $V(C)\cap\{u,v\}=\emptyset$, then $C$ is also a cycle in $G/e$, and in case $|V(C)\cap\{u,v\}|=1$, then by replacing the vertex in $V(C)\cap\{u,v\}$ by $x_{\{u,v\}}$ in $C$, we obtain a cycle of the same length as $C$ in $G/e$. In both of these cases, the proof is complete, and thus we next suppose that $\{u,v\}\subseteq V(C)$.

Let us denote $C=v_1-v_2-v_3-\cdots-v_\ell-v_1$, where $v_1=u$ and $v_i=v$ for some $i\in[\ell]\setminus\{1\}$. Note that $\ell\geq 2k$. Without loss of generality, assume that $i-2\geq \ell-i$ (else we replace each $v_j$, except for $v_1$, by $v_{\ell-j}$, and obtain a cycle where this property holds). Now, note that $C'=x_{\{u,v\}}-v_2-\cdots-v_{i-1}-x_{\{u,v\}}$ is a cycle in $G/e$. Moreover, since $i-2\geq \ell-i$, it holds that $|V(C')|=i-1\geq\frac{\ell}{2}\geq k$. Thus, $(G/e,f_{/e},k)$ is a \Yes-instance.
\end{proof}

Before we present our algorithm, we need two additional propositions, handling the extreme cases where we either discover that our input graph contains a large grid or, after a series of operations that contracted edges in $G$, we ended up with a graph isomorhpic to the cell graph of $G$. For the first case, we need the following result (see also \cite{Demaine:2008mi,cygan2015parameterized}).

\begin{observation}\label{obs:longCycLargeGrid}
Let $(G,k)$ be an instance of \probKCycle{} on general graphs. If $G$ contains a $\sqrt{k}\times \sqrt{k}$ grid as a minor, then $(G,k)$ is a \Yes-instance.
\end{observation}

For the second case, we need the following result (see also \cite{cygan2015parameterized}).

\begin{proposition}[\cite{BodlaenderCK12}]\label{prop:DPLongCyc}
\probKCycle on graphs of treewidth $\mathrm{\it tw}$ can be solved in time $2^{\OO(\mathrm{\it tw})}\cdot n^{\OO(1)}$.
\end{proposition}

From Proposition \ref{prop:DPLongCyc} and the fact that $\cell(G)$ is a minor of $G$, we have the following.

\begin{observation}\label{obs:longCycEndCell}
Let $(G,f,k)$ be an instance of \probKCycle on \gcgraphs{}. Then, it can be determined in time $2^{\OO(\tw(\cell(G)))}\cdot n^{\OO(1)}$ whether $\cell(G)$ contains a cycle on at least $k$ vertices, in which case $(G,f,k)$ is a \Yes-instance.
\end{observation}

We are now ready to present our algorithm. The proof of correctness and analysis of running times are integrated in the description of the algorithm.

\subparagraph*{Proof of Theorem \ref{thm:longestCyc}.}
Let $(G,O,k)$ be an instance of \probKCycle on unit disk/square graphs. By using Corollary~\ref{cor:firstPhaseToGrid}, we first obtain an equivalent instance $(G,f,k)$ of \probKCycle on \gcgraphs{} (both instances refer to the same graph $G$ and parameter $k$). Next, by using Lemma~\ref{lem:mainCell} with the parameter $\ell=100\cdot599^3\cdot\sqrt{k}$, we either correctly conclude that $G$ contains a $\sqrt{k}\times\sqrt{k}$ grid as a minor, or compute a tree decomposition of $\cell(G)$ of width at most $500\cdot599^3\cdot\sqrt{k}=\OO(\sqrt{k})$. In the first case, by Observation \ref{obs:longCycLargeGrid}, we are done. In the latter case, by using Observation~\ref{obs:longCycEndCell}, we determine in time $2^{\OO(\sqrt{k})}\cdot n^{\OO(1)}$ whether $\cell(G)$ contains a cycle on at least $k$ vertices, where if the answer is positive, then we are done. Thus, we next suppose that $\cell(G)$ does not contain a cycle on at least $k$ vertices.

Now, as long as there exists a contractible pair $(u,v)$, we perform the following operation. First, by using Corollary~\ref{cor:nearKCycle}, we determine in time $2^{\OO(\sqrt{k}\log k)}\cdot n^{\OO(1)}$ whether $G$ contain a cycle whose number of vertices is between $k$ and $2k$. If the answer is positive, then we are done (our final answer is \Yes). If the answer is negative, then we contract the edge $\{u,v\}$. By Lemma \ref{lem:longCycPreserveYes}, we obtain an instance that is equivalent to the previous one. Note that the loop described in this paragraph can have at most $\OO(n^2)$ iterations, and therefore its total running time is bounded by $2^{\OO(\sqrt{k}\log k)}\cdot n^{\OO(1)}$.

Once there does not exist a contractible pair $(u,v)$, as we have only modified the graph by contracting edges, on at a time, between contractible pairs, we are left with a graph that is isomorphic to the cell graph of our original input graph. We have already correctly concluded that this graph does not contain a cycle on at least $k$ vertices. Thus, at this point, we correctly answer \No.\qed

%!TEX root = main_unit_disk.tex

\section{Feedback Vertex Set}\label{sec:fvs}

In this section, we show that \probFVS{} admits a subexponential-time parameterized algorithm. More precisely, we prove the following.

\begin{theorem}\label{thm:fvs}
\probFVS{} on unit disk/square graphs can be solved in time $2^{\OO(\sqrt{k}\log k)}\cdot n^{\OO(1)}$.
\end{theorem}

%The arguments to prove Theorem~\ref{thm:fvs} are like the one in Section~{sec:cycPack}
First, we observe that if we find a large grid, we can answer \No\ (see also \cite{Demaine:2008mi,cygan2015parameterized}).

\begin{observation}\label{obs:fvsLargeGrid}
Let $(G,k)$ be an instance of  \probFVS. If $G$ contains a $2\sqrt{k}\times 2\sqrt{k}$ grid as a minor, then $(G,k)$ is a \No-instance.
\end{observation}

This observation leads us to the following.

\begin{lemma}\label{lem:fvsFirstPhase}
Let $(G,O,k)$ be an instance of  \probFVS{} on unit disk/square graphs. Then, in time $2^{\OO(\sqrt{k}\log k)}\cdot \vert V(G)\vert^{\OO(1)}$, one can either solve $(G,O,k)$ or obtain an equivalent instance $(G,f,k)$ of \probFVS{} on \gcgraphs{} together with an \NCTD{$\OO(\sqrt{k})$} of $G$.
\end{lemma}

\begin{proof}
First, by using Lemmata \ref{lem:unitDisk} or \ref{lem:unitSquare}, we obtain a representation $f$ of $G$. Then, by using Corollary \ref{cor:mainNCTD} with $\ell=200\cdot599^3\cdot\sqrt{k}=\OO(\sqrt{k})$, we either correctly conclude that $G$ contains a $2\sqrt{k}\times 2\sqrt{k}$ grid as a minor, or compute an \NCTD{$\OO(\sqrt{k})$} of $G$. In both cases, by Observation \ref{obs:fvsLargeGrid}, we are done.
\end{proof}

Because of  Lemma~\ref{lem:fvsFirstPhase}, to prove Theorem~\ref{thm:fvs}, we can focus on 
\probFVS{} on \gcgraph{s}, where the input also contains a   \NCTD{$\OO(\sqrt{k})$}. 
%Let $(G,f,k,{\cal T})$ be an instance of \probFVS, where $G$ is a \gcgraph{} with representation $f$ and ${\cal T}$ is a 
%\NCTD{$c\sqrt{k}$} for some constant $c$. . 
That is, the input of \probFVS{} on \gcgraph{s} is a tuple $(G,f,k,{\cal T})$ 
where $G$ is a \gcgraph{} with representation $f$ and ${\cal T}=(T,\beta)$ is a 
\NCTD{$\OO(\sqrt{k})$} of $G$. 
Notice that if there is a cell $(i,j)$ of $f$, such that $\vert f^{-1}(i,j)\vert \geq k+3$, 
then there is no feedback vertex set of size at most $k$ in $G$, because 
$f^{-1}(i,j)$ is a clique of size at least $k+3$ in $G$. 

\begin{observation}
\label{obs:fvslowcell}
Let $(G,f,k,{\cal T})$ be an instance of \probFVS, where $G$ is a \gcgraph{} with representation $f$. 
If there is a cell $(i,j)$ in $f$ such that $\vert f^{-1}(i,j)\vert \geq k+3$, then  $(G,f,k,{\cal T})$ is a \No-instance. 
\end{observation} 

The following observation follows from the fact that ${\cal T}$ is a \NCTD{$\OO(\sqrt{k})$} and $\vert f^{-1}(i,j)\vert \leq k+2$ for any cell $(i,j)$ of $f$. 
\begin{observation}
\label{obs:fvs2}
For any $v\in V(T)$, $\vert \beta(v)  \vert = \OO(k^{1.5})$. 
\end{observation}

Notice that $G$ has a feedback vertex set of size at most $k$ if and only if there is a vertex subset $F \subseteq V(G)$ 
of cardinality at least $\vert V(G)\vert -k$ such that $G[F]$ is a forest. 
Hence, instead of stating the problem as finding a $k$ sized feedback vertex set in $G$, we can state it  as finding  an induced subgraph $H$ of $G$ with maximum number of vertices such that $H$ is a forest.    

\defparproblem{\maxforestlong~ (\maxforest)}
{
A \gcgraph{} $G$ with representation $f$ and an integer $k$ such that 
${\cal T}$ is a \NCTD{$c\sqrt k$} of $G$ and for any cell $(i,j)$ in $f$, $\vert f^{-1}(i,j)\vert \leq k+2$, 
where $c$ is a constant} {$k$}{Is there subset $W\subseteq V(G)$ such that $G[W]$ is a forest and 
$\vert W \vert \geq \vert V(G)\vert -k$}

\begin{observation}
\label{obs:forest}
Let $(G,f,k,{\cal T})$ be an instance of \maxforest. 
Then $(G,f,k,{\cal T})$ is a \Yes-instance of \maxforest{} if and only if $(G,f,k,{\cal T})$ is 
a \Yes-instance of \probFVS. 
\end{observation}

By Lemma~\ref{lem:fvsFirstPhase} and Observations~\ref{obs:fvslowcell} and \ref{obs:forest}, to prove Theorem~\ref{thm:fvs}, it is sufficient that we prove the following result (which is the focus of the rest of this section).

\begin{lemma}\label{lem:mainfvs}
\maxforest{} on \gcgraphs{} 
%where the input includes an \NCTD{$\OO(\sqrt{k})$} of input graph $G$, 
can be solved in time $2^{\OO(\sqrt{k}\log k)}\cdot n^{\OO(1)}$.
\end{lemma}
\begin{proof}[Proof sketch]

%So, now on we have that for any cell $(i,j)$ of $f$,  $\vert f^{-1}(i,j)\vert \leq k+2$. 

We explain a DP algorithm which given as input 
$(G,f,k,{\cal T})$ where $G$ is a \gcgraph{} with representation $f$, ${\cal T}=(T,\beta)$ is a 
 \NCTD{$c\sqrt{k}$},  $c$ is a constant and $\vert f^{-1}(i,j)\vert \leq k+2$ for any cell $(i,j)$ of $f$ and 
outputs \Yes{} if there is an induced forest with at least $\vert V(G)\vert-k$ vertices and outputs \No{} otherwise. 
Here we use the term solution for  a vertex subset $S\subseteq V(G)$ with the property  that $G[S]$ is a forest.    
First notice that any solution $S$ contains at most $2$ vertices from $f^{-1}(i,j)$ for any cell $(i,j)$. 
Now, the following claim follows from the fact that ${\cal T}$ is a \NCTD{$c\sqrt{k}$} 
and any solution contain at most $2$ vertices from $f^{-1}(i,j)$ for any cell $(i,j)$. 
\begin{claim}
\label{claim:fvs1}
For any $v\in V(T)$ and any solution $S$, $\vert S\cap \beta(v) \vert \leq 2c\sqrt k$. 
\end{claim}

We first briefly explain what is the table entries in a standard DP 
algorithm for our problem on graphs of bounded treewidth~\cite{cygan2015parameterized}. 
%\todo{similar problems are explained and FVS is given as example without proof}. 
Then we explain that in fact many of the entries we compute in 
the standard DP table is redundant in our case, because of Observation~\ref{obs:fvs2} and Claim~\ref{claim:fvs1}.   
That is, Observation~\ref{obs:fvs2} and Claim~\ref{claim:fvs1}, shows that only $2^{\OO(\sqrt k \log k)} \vert V(G)\vert ^{\OO(1)}$ many states in the DP table are relevant in our case. 
Recall that for any $v\in V(T)$, $\gamma(v)$ denote the union of the bags of $v$ and its descendants.
%
%Let $t$ be a node of 
%$T$ and $X_t$ be the bag associated with it. We need the following definitions.  
%We use $G_t$ to denote the graph induced  by the
%vertex set  $\bigcup_{t'}X_{t'}$, where $t'$ ranges over all descendants of $t$,
%including $t$. 
%
The standard DP table for our problem will have an entry indexed 
by $(v,U, U_1\uplus U_2\ldots U_{\ell}=U)$ where $v\in V(T)$, $U\subseteq \beta(v)$. 
The table entry ${\cal A}[v,U, U_1\uplus U_2\ldots U_{\ell}]$ stores the following 
information: the maximum cardinality of a vertex subset $W\subseteq G[\gamma(v)]$ such 
that $W\cap \beta(v)=U$, $G[W]$ is a forest with a set of connected components $\CC$ and for 
any $C\in \CC$, either $V(C)\cap \beta(v)=\emptyset$ or $V(C)\cap \beta(v)= U_i$ for some $i\in [\ell]$.  
Notice that the total number of DP table entries is bounded by $\tw^{\OO(\tw)}\vert V(G)\vert ^{\OO(1)}$ 
where $\tw$ is the width of the tree decomposition ${\cal T}$. One can easily show that the 
computation of the DP table at a  node can be done in time polynomial in the size of the tables of its children. 

By Observation~\ref{obs:fvs2} and Claim~\ref{claim:fvs1}, we know that for any bag $\beta(v)$ in ${\cal T}$, the potential number 
of subsets of $\beta(v)$ which can be part of any solution is at most $2^{\OO(\sqrt k \log k)}$.  
This implies that we only need to compute the DP table entries for indices 
$(v,U, U_1\uplus U_2\ldots U_{\ell}=U)$ where $v\in V(T)$, $U\subseteq \beta(v)$ and $\vert U \vert \leq 2c\sqrt k$. 
Thus, the size of DP table, and hence the time to compute it takes   $2^{\OO(\sqrt k \log k)} n^{\OO(1)}$ time. This concludes the description. 
\end{proof}

%Let $\TT$ denote our \NCTD{$\OO(\sqrt{k})$} of $G$. We proceed by considering the ``interaction'' between cells in the context of the manner in which cycles in a solution cross their boundaries. More precisely, we define the following notion.
%
%\begin{definition}
%A set $\cal C$ of pairwise disjoint is {\em simple} if it satisfies the following. \hly{TODO}
%\end{definition}
%
%Next, we show that we can focus on the deciding whether a simple set, rather than a general set, of $k$ pairwise-disjoint cycles exists.
%
%\begin{lemma}\label{lem:cycPackBoudInteract}
%If $(G,f,k)$ is a \Yes-instance, then there exists a simple set $\cal C$ of $k$ pairwise-disjoint cycles.
%\end{lemma}
%
%\begin{proof}
%Suppose that $(G,f,k)$ is a \Yes-instance.
%\hly{TODO}
%\end{proof}
%
%To proceed, we recall the standard procedure, based on DP, that solves {\sc Cycle Packing}. Proofs of correctness of such procedures can be found in textbooks such as \cite{cygan2015parameterized}. For completeness, we will explain why this procedure runs in time $2^{\OO(k\log k)}\cdot n^{\OO(1)}$. Then, we will show how Lemma~\ref{lem:cycPackBoudInteract} can be used to speed-up this procedure to run in time $2^{\OO(\sqrt{k}\log k)}\cdot n^{\OO(1)}$. This would conclude the proof of Lemma~\ref{lem:mainCycPack} (which concludes the proof of Theorem~\ref{thm:cycPack}).

%!TEX root = main_unit_disk.tex

\section{Cycle Packing}\label{sec:cycPack}

In this section, we show that {\sc Cycle Packing} admits a subexponential-time parameterized algorithm. More precisely, we prove the following.

\begin{theorem}\label{thm:cycPack}
{\sc Cycle Packing} on unit disk/square graphs can be solved in time $2^{\OO(\sqrt{k}\log k)}\cdot n^{\OO(1)}$.
\end{theorem}

First, we observe that if we find a large grid, we can answer \Yes{} (see also \cite{Demaine:2008mi,cygan2015parameterized}).

\begin{observation}\label{obs:cycPackLargeGrid}
Let $(G,k)$ be an instance of {\sc Cycle Packing} on general graphs. If $G$ contains a $2\sqrt{k}\times 2\sqrt{k}$ grid as a minor, then $(G,k)$ is a \Yes-instance.
\end{observation}

This observation leads us to the following.

\begin{lemma}\label{lem:cycPackFirstPhase}
Let $(G,O,k)$ be an instance of {\sc Cycle Packing} on unit disk/square graphs. Then, in time $2^{\OO(\sqrt{k}\log k)}\cdot n^{\OO(1)}$, one can either solve $(G,O,k)$ or obtain an equivalent instance $(G,f,k)$ of {\sc Cycle Packing} on \gcgraphs{} together with an \NCTD{$\OO(\sqrt{k})$} of $G$.
\end{lemma}

\begin{proof}
First, by using Lemmata \ref{lem:unitDisk} or \ref{lem:unitSquare}, we obtain a representation $f$ of $G$. Then, by using Corollary \ref{cor:mainNCTD} with $\ell=200\cdot599^3\cdot\sqrt{k}=\OO(\sqrt{k})$, we either correctly conclude that $G$ contains a $2\sqrt{k}\times 2\sqrt{k}$ grid as a minor, or compute an \NCTD{$\OO(\sqrt{k})$} of $G$. In both cases, by Observation \ref{obs:cycPackLargeGrid}, we are done.
\end{proof}

Now, note that if there exists a cell $(i,j)\in[t]\times[t']$ such that $|f^{-1}(i,j)|\geq 3k$, then by Condition \ref{condition:GridClique1} in Definition \ref{def:GridClique}, $G[f^{-1}(i,j)]$ is a clique on at least $3k$ vertices and thus it contains $k$ pairwise vertex-disjoint cycles (triangles). More precisely, we have the following.

\begin{observation}\label{obs:smallCliqueCycPack}
Let $(G,f: V(G)\rightarrow[t]\times[t'],k)$ of {\sc Cycle Packing} on \gcgraphs. Then, if there exists a cell $(i,j)\in[t]\times[t']$ such that $|f^{-1}(i,j)|\geq 3k$, then $(G,f,k)$ is a \Yes-instance.
\end{observation}

By Lemma~\ref{lem:cycPackFirstPhase} and Observation~\ref{obs:smallCliqueCycPack}, to prove Theorem~\ref{thm:cycPack}, it is sufficient that we prove the following result (which is the focus of the rest of this section).

\begin{lemma}\label{lem:mainCycPack}
{\sc Cycle Packing} on \gcgraphs{} can be solved in time $2^{\OO(\sqrt{k}\log k)}\cdot n^{\OO(1)}$, assuming that the input includes an \NCTD{$\OO(\sqrt{k})$} of $G$, and that for every cell $(i,j)\in[t]\times[t']$, it holds that $|f^{-1}(i,j)|\leq 3k$.
\end{lemma}

Let $(G,f:V(G)\rightarrow[t]\times[t'],k)$ denote the input instance of {\sc Cycle Packing}, and let ${\TT}=(T,\beta)$ denote our \NCTD{$\OO(\sqrt{k})$} of $G$. Note that since ${\TT}$ is an \NCTD{$\OO(\sqrt{k})$}, and since for every $(i,j)\in[t]\times[t']$, it holds that $|f^{-1}(i,j)|\leq 3k$, we also have the following.

\begin{observation}\label{obs:sizeBagCycPack}
For all $v\in V(T)$, it holds that $|\beta(v)|=\OO(k^{1.5})$.
\end{observation}

We proceed by considering the ``interaction'' between cells in the context of the manner in which cycles in a solution cross their boundaries. To be precise, by Definition \ref{def:GridClique}, we first observe the following.

\begin{observation}\label{obs:inducedCyc}
Let $C$ be an induced cycle in $G$. Then, there does not exist a cell $(i,j)\in[t]\times[t']$ and two distinct vertices $u,v\in V(C)\cap f^{-1}(i,j)$ such that $\{u,v\}\notin E(C)$. In particular, for every cell $(i,j)\in[t]\times[t']$, exactly one of the following conditions holds.
\begin{enumerate}
\item $V(C)\subseteq f^{-1}(i,j)$.
\item $|V(C)\cap f^{-1}(i,j)|=1$.
\item $|V(C)\cap f^{-1}(i,j)|=2$ and the two vertices in $V(C)\cap f^{-1}(i,j)$ are neighbors in $C$.
\end{enumerate}
\end{observation}

Next, note that given a set ${\cal C}$ of pairwise vertex-disjoint cycles and a cycle $C\in{\cal C}$ that is not an induced cycle in $G$, by replacing $C$ in ${\cal C}$ by an induced cycle in $G[V(C)]$, we obtain another set of pairwise vertex-disjoint cycles. Thus, we have the following.

\begin{observation}\label{obs:findInducedCycs}
If $(G,f,k)$ is a \Yes-instance, then $G$ contains a set $\cal C$ of $k$ pairwise-disjoint {\em induced} cycles.
\end{observation}

\begin{definition}
Given two distinct cells $(i,j),(i',j')\in [t]\times[t']$, we say that $C$ {\em crosses} $((i,j),(i',j'))$ if there exist (not necessarily distinct) $u,v\in f^{-1}(i,j)$ and distinct $w,r\in f^{-1}(i',j')$ such that $\{u,w\},\{v,r\}\in E(C)$. Moreover, we say that $C$ {\em crosses} $\{(i,j),(i',j')\}$ if it crosses at least one of the pairs $((i,j),(i',j'))$ and $((i',j'),(i,j))$.
\end{definition}

\begin{definition}
Given three distinct cells $(i_1,j_1),(i_2,j_2),(i_3,j_3)\in [t]\times[t']$, we say that $C$ {\em crosses} $((i_1,j_1),(i_2,j_2),(i_3,j_3))$ if there exist $u\in f^{-1}(i_1,j_1)$, (not necessarily distinct) $v,w\in f^{-1}(i_2,j_2)$ and $r\in f^{-1}(i_3,j_3)$ such that $\{u,v\},\{w,r\}\in E(C)$. Moreover, we say that $C$ {\em crosses} $\{(i_1,j_1),(i_2,j_2),(i_3,j_3)\}$ if it crosses at least one of the tuples in $\{((i_s,j_s),(i_r,j_r),(i_t,j_t)): \{s,r,t\}=\{1,2,3\}\}$.
\end{definition}

Next, we use the definitions above to capture the set of cycles which we would like detect (if a solution exists).

\begin{definition}
A set $\cal C$ of pairwise vertex-disjoint induced cycles is {\em simple} if it satisfies the following conditions.
\begin{itemize}
\item For every two distinct cells $(i,j),(i',j')\in [t]\times[t']$, there exist at most two cycles in $\cal C$ that cross $\{(i,j),(i',j')\}$.
\item For every three distinct cells $(i_1,j_1),(i_2,j_2),(i_3,j_3)\in [t]\times[t']$, there exist at most two cycles in $\cal C$ that cross $\{(i_1,j_1),(i_2,j_2),(i_3,j_3)\}$.
\end{itemize}
\end{definition}

Given a cycle $C$ (in $G$), denote cross$(C)=\{\{u,v\}\in E(C): f(u)\neq f(v)\}$, and given a set $\cal C$ of cycles, denote cross$({\cal C})=\displaystyle{\cup_{C\in{\cal C}}\mathrm{cross}(C)}$. Next, we show that we can focus on the deciding whether a simple set, rather than a general set, of $k$ pairwise-disjoint cycles exists.

\begin{lemma}\label{lem:cycPackBoudInteract}
If $(G,f,k)$ is a \Yes-instance, then $G$ contains a simple set of $k$ pairwise-disjoint induced cycles.
\end{lemma}

\begin{proof}
Suppose that $(G,f: V(G)\rightarrow[t]\times[t'],k)$ is a \Yes-instance. Next, let $\cal C$ denote a set of $k$ pairwise vertex-disjoint induced cycles that minimizes $|\mathrm{cross}({\cal C})|$ among all such sets of cycles (the existence of at least one such set of induced cycles is guaranteed by Observation~\ref{obs:findInducedCycs}). We will show that $\cal C$ is simple. In what follows, we implicitly rely on Condition \ref{condition:GridClique1} in Definition~\ref{def:GridClique}.

First, suppose that there exist two distinct cells $(i,j),(i',j')\in [t]\times[t']$ and three cycles in $\cal C$ that cross $\{(i,j),(i',j')\}$. Let $C_1,C_2$ and $C_3$ denote these three cycles. Then, at least one of the three following conditions is true.
\begin{enumerate}
\item There exist distinct $s,t\in\{1,2,3\}$ such that $|(V(C_s)\cup V(C_t))\cap f^{-1}(i,j)|\geq 3$ and $|(V(C_s)\cup V(C_t))\cap f^{-1}(i',j')|\geq 3$: In this case, we replace $C_s$ and $C_t$ in ${\cal C}$ by some cycle on three vertices in $G[(V(C_s)\cup V(C_t))\cap f^{-1}(i,j)]$ and some cycle on three vertices in $G[(V(C_s)\cup V(C_t))\cap f^{-1}(i',j')]$. We thus obtain a set of $k$ pairwise vertex-disjoint cycles, ${\cal C}'$, such that $|\mathrm{cross}({\cal C}')|<|\mathrm{cross}({\cal C})|$, which is a contradiction to the choice of $\cal C$.

\item $|(V(C_1)\cup V(C_2)\cup V(C_3))\cap f^{-1}(i,j)|=3$ and $|(V(C_1)\cup V(C_2)\cup V(C_3))\cap f^{-1}(i',j')|\geq 6$: In this case, we replace $C_1, C_2$ and $C_3$ in ${\cal C}$ by some cycle on three vertices in $G[(V(C_1)\cup V(C_2)\cup V(C_3))\cap f^{-1}(i,j)]$ and two vertex-disjoint cycles, each on three vertices, in $G[(V(C_1)\cup V(C_2)\cup V(C_3))\cap f^{-1}(i',j')]$. We thus obtain a set of $k$ pairwise vertex-disjoint cycles, ${\cal C}'$, such that $|\mathrm{cross}({\cal C}')|<|\mathrm{cross}({\cal C})|$, which is a contradiction to the choice of $\cal C$.

\item $|(V(C_1)\cup V(C_2)\cup V(C_3))\cap f^{-1}(i,j)|\geq 6$ and $|(V(C_1)\cup V(C_2)\cup V(C_3))\cap f^{-1}(i',j')|=3$: This case is symmetric to the previous one.
\end{enumerate}

Second, suppose that there exist three distinct cells $(i_1,j_1),(i_2,j_2),(i_3,j_3)\in [t]\times[t']$ and three cycles in $\cal C$ that cross $\{(i_1,j_1),(i_2,j_2),(i_3,j_3)\}$. Let $C_1,C_2$ and $C_3$ denote these three cycles. Then, it holds that $|(V(C_1)\cup V(C_2)\cup V(C_3))\cap f^{-1}(i_1,j_1)|\geq 3$, $|(V(C_1)\cup V(C_2)\cup V(C_3))\cap f^{-1}(i_2,j_2)|\geq 3$ and $|(V(C_1)\cup V(C_2)\cup V(C_3))\cap f^{-1}(i_3,j_3)|\geq 3$. We replace $C_1, C_2$ and $C_3$ in ${\cal C}$ by some cycle on three vertices in $G[(V(C_1)\cup V(C_2)\cup V(C_3))\cap f^{-1}(i_1,j_1)]$, some cycle on three vertices in $G[(V(C_1)\cup V(C_2)\cup V(C_3))\cap f^{-1}(i_2,j_2)]$, and some cycle on three vertices in $G[(V(C_1)\cup V(C_2)\cup V(C_3))\cap f^{-1}(i_3,j_3)]$.
\end{proof}

Now, we examine the information given by Lemma~\ref{lem:cycPackBoudInteract} to extract a form that will be easier for us to exploit. To this end, we need the following. Given a set $\cal C$ of cycles and a cell $(i,j)\in[t]\times[t']$, denote cross$({\cal C},i,j) = (\bigcup\mathrm{cross}({\cal C}))\cap f^{-1}(i,j)$ (note that $\bigcup\mathrm{cross}({\cal C})$ is the set of every vertex that is an endpoint of an edge in $\mathrm{cross}({\cal C})$).

\begin{lemma}\label{lem:boundedIntercationFin}
If $(G,f: V(G)\rightarrow[t]\times[t'],k)$ is a \Yes-instance, then $G$ contains a set $\cal C$ of $k$ pairwise-disjoint induced cycles such that for every cell $(i,j)\in[t]\times[t']$, it holds that $|\mathrm{cross}({\cal C},i,j)|\leq 2304=\OO(1)$.
\end{lemma}

\begin{proof}
Suppose that $(G,f,k)$ is a \Yes-instance. By Lemma~\ref{lem:cycPackBoudInteract}, there exists a simple set $\cal C$ of $k$ pairwise-disjoint induced cycles. Let $(i,j)\in[t]\times[t']$ be some cell. Given a cell $(i',j')\in[t]\times[t']\setminus\{(i,j)\}$, denote ${\cal C}(i',j')=\{C\in{\cal C}: C$ crosses $\{(i,j),(i',j')\}\}$. Moreover, given cells $(i_2,j_2),(i_3,j_3)\in[t]\times[t']\setminus\{(i,j)\}$, denote ${\cal C}(i_2,j_2,i_3,j_3)=\{C\in{\cal C}: C$ crosses $\{(i,j),(i_2,j_2),(i_3,j_3)\}\}$. Now, we define two sets of indices:
\begin{itemize}
\item ${\cal I}=\{(i',j')\in[t]\times[t']\setminus\{(i,j)\}: {\cal C}(i',j')\neq\emptyset\}$.
\item ${\cal I}'=\{((i_2,j_2),(i_3,j_3))\in([t]\times[t']\setminus\{(i,j)\})\times([t]\times[t']\setminus\{(i,j),(i_2,j_2)\}): {\cal C}(i_2,j_2,i_3,j_3)\neq\emptyset\}$.
\end{itemize}

Then, by Observation \ref{obs:inducedCyc} and Lemma~\ref{lem:cycPackBoudInteract}, it holds that
\[\begin{array}{ll}
|\mathrm{cross}({\cal C},i,j)| &\leq \displaystyle{2(|\bigcup_{(i',j')\in{\cal I}}{\cal C}(i',j')|+|\bigcup_{((i_2,j_2),(i_3,j_3))\in{\cal I}'}{\cal C}(i_2,j_2,i_3,j_3)|)}\\
& \leq 4(|{\cal I}|+|{\cal I}'|).
\end{array}\]

Note that $|\{(i',j')\in[t]\times[t']\setminus\{(i,j)\}~|~|i-i'|\leq 2,|j-j'|\leq 2\}|\leq 24$. Thus, by Condition \ref{condition:GridClique2} in Definition \ref{def:GridClique}, we have that $|{\cal I}|\leq 24$ and $|{\cal I}'|\leq 24\cdot 23=552$. Therefore, $\mathrm{cross}({\cal C},i,j)|\leq 4(24+552)=2304$.
\end{proof}

We are now ready to prove Lemma~\ref{lem:mainCycPack}. Except for the arguments where we crucially rely on Lemma~\ref{lem:boundedIntercationFin} and the fact the we have an $\OO(\sqrt{k})$-NCTD and not some general nice tree decomposition, the description of the DP is standard (see, e.g., \cite{cygan2015parameterized}). Thus, we only give a sketch of the proof. 

\subparagraph*{Proof sketch of Lemma~\ref{lem:mainCycPack}.} Let us first examine a standard DP table $\cal A$ to solve {\sc Cycle Packing} when the parameter is $\tw(G)$. Here, we have an entry ${\cal A}[v,{\cal Z},k']$ for every node $v\in V(T)$, multiset $\cal Z$ of subsets of sizes 1 or 2 of $\beta(v)$ and nonnegative integer $k'\leq k$. Moreover, each set of size 1 in $\cal Z$ has only one occurrence and its vertex does not appear in any set of size 2 in $\cal Z$, and every vertex in $\beta(v)$ appears in at most two sets in $\cal Z$.
Each such entry stores either 0 or 1. The value is 1 if an only if there exist a set ${\cal S}$ of $k'$ pairwise vertex-disjoint cycles in $G[\gamma(v)]$ and a set ${\cal P}$ of internally pairwise vertex-disjoint paths in $G[\gamma(v)]$ such that the following conditions are satisfied.
\begin{itemize}
\item $(\bigcup_{C\in{\cal S}}V(C))\cap(\bigcup_{P\in{\cal P}}V(P))=\emptyset$.
\item On the one hand, for every cycle $C\in{\cal C}$, it holds that $|V(C)\cap \beta(v)|\leq 1$ and if $|V(C)\cap \beta(v)|=1$ then there exists a set in $\cal Z$ that is equal to $V(C)\cap \beta(v)$. On the other hand, if $\cal Z$ contains a set of size 1, then there exists a cycle $C\in {\cal C}$ such that $V(C)\cap\beta(v)$ equals this set.
\item On the one hand, for every path $P\in{\cal P}$, it holds that $P$ contains at least three vertices, both endpoints of $P$ belong to a distinct occurrence of a set in $\cal Z$ (of size 2), and none of the internal vertices of $P$ belongs to $\beta(v)$. On the other hand, for every occurrence $X$ of a set of size 2 in $\cal Z$, there exists a distinct path $P$ in $\cal P$ such that the set containing the two endpoints of $P$ is equal to $X$.
\end{itemize}

The entry ${\cal A}[v,{\cal Z},k']$ can be computed by examining the all entries ${\cal A}[u,\widehat{\cal Z},\widehat{k}]$ where $u$ is a child of $v$ in $T$ (recall that $v$ can have at most two children). At the end of the computation of $\cal A$, we conclude that the input instance is a \Yes-instance if and only if ${\cal A}[r,\emptyset,k]$ contains 1 where $r$ is the root of $T$. By Observation~\ref{obs:sizeBagCycPack}, we deduce that $\cal A$ contains $2^{\OO(k\log k)}\cdot n$ entries, where each entry can be computed in time $2^{\OO(k\log k)}$.

We claim that for every $v\in V(T)$, it is sufficient to compute only $2^{\OO(\sqrt{k}\log k)}$ entries. More precisely, for every $v\in V(T)$, it is sufficient to compute only entries ${\cal A}[v,{\cal Z},k']$ such that $|\bigcup{\cal Z}|=\OO(\sqrt{k})$ (there are only $2^{\OO(\sqrt{k}\log k)}$ such entries). Indeed, suppose that the input instance is a \Yes-instance. Then, by Lemma~\ref{lem:boundedIntercationFin}, there exists a set $\cal C$ of $k$ pairwise vertex-disjoint induced cycles such that for every cell $(i,j)\in[t]\times[t']$, it holds that $|\mathrm{cross}({\cal C},i,j)|=\OO(1)$. Now, we sketch the main arguments that show that for every $v\in V(T)$, we still have an entry that ``captures'' $\cal C$ (as explained below) and we are able to compute it in time $2^{\OO(\sqrt{k}\log k)}$, which would imply that eventually, we would still be able to deduce that ${\cal A}[r,\emptyset,k]$ contains 1. For this purpose, consider some $v\in V(T)$. First, we notice that since for every cell $(i,j)\in[t]\times[t']$, it holds that $|\mathrm{cross}({\cal C},i,j)|=\OO(1)$, by Observation~\ref{obs:inducedCyc}, and since $\TT$ is an \NCTD{$\OO(\sqrt{k})$}, we have that there exists a set $U$ of at most $\OO(\sqrt{k})$ vertices in $\beta(v)$ such that every cycle $C\in{\cal C}$ satisfies at least one of the following conditions.
\begin{enumerate}
\item $V(C)\cap \beta(v)\subseteq U$
\item $V(C)\subseteq\gamma(v)\setminus\beta(v)$.
\item $V(C)\subseteq V(G)\setminus\gamma(v)$.
\end{enumerate}

Now, we let $\cal S$ denote the set of cycles in $\cal C$ such that all of their vertices, except at most one that belongs to $\beta(v)$, belong to $\gamma(v)\setminus\beta(v)$. Accordingly, we denote $k'=|{\cal S}|$. Moreover, let $\cal P$ denote the set of every subpath of a cycle in $\cal C$ whose endpoints belong to $\beta(v)$ and whose set of internal vertices is a subset of size at least 1 of $\gamma(v)\setminus\beta(v)$. Finally, we define ${\cal Z}$ as the multiset $\{\beta(v)\cap O~|~O\in {\cal S}\cup{\cal P}\}$. Then, it holds that $|\bigcup{\cal Z}|=\OO(\sqrt{k})$  and $\cal C$ witnesses that ${\cal A}[v,{\cal Z},k']$ should be 1. Overall, by the existence of the set $U$ that is mentioned above, we conclude the entry ${\cal A}[v,{\cal Z},k']$ can be computed in time $2^{\OO(\sqrt{k}\log k)}$. This completes the proof sketch.\qed

\section{Conclusion}\label{sec:conclusion}
In this paper, we gave  subexponential algorithms of running time $2^{\cO({\sqrt{k}\log{k}})} \cdot n^{\cO(1)}$ for a number of parameterized problems about cycles in unit disk graphs. The first natural question is whether the $\log{k}$ factor in the exponent can be shaved off. While we were not able to do it, we do not exclude such a possibility. 
%More importantly, 
In particular, it would be very interesting to build a theory  for unit disk graphs, which is similar to the bidimensionality theory for planar graphs. In this context, it will be useful to provide a general characterization of parameterized problems admitting subexponential algorithms on unit disk graphs. 

\bibliography{book_kernels_fvf,refs}

%\newpage
%\section{Appendix: Problem Definitions}
%\label{sec:probdefn}
%\todo[inline]{@FAHAD: Please formally define the problems here: Longest Path, Cycle, Exact Cycle, Subgraph Isomorphism, Cycle Packing, Connected Vertex Cover}
%\defparproblem{}
%{An unit disk graph  $G$ with the representation} {$k$}{Is there subset $W\subseteq V(G)$ such that $G[W]$ is a forest and 
%$\vert W \vert \geq \vert V(G)\vert -k$}

\end{document}